\documentclass[a4paper,11pt]{article}
\usepackage{latexsym,amssymb,amsfonts,amsmath,amsthm}
\usepackage[dvips]{graphicx}
\usepackage{authblk}

\newtheorem{theorem}{Theorem}
\newtheorem{lemma}{Lemma}
\newtheorem{corollary}{Corollary}
\newtheorem{proposition}{Proposition}

\newtheorem{definition}{Definition}

\usepackage[font={small,it}]{caption}

\usepackage{color}
\newcommand{\red}{\textcolor{black}}

\title{The 2-Tessellable Quantum Walk as \\
a Quantum Markov Chain}
\title{Eigenbasis of the Evolution Operator of\\ 2-Tessellable Quantum Walks}
\author{ 
Yusuke Higuchi,$^{1}$ 
Renato Portugal,$^{2}$ 
Iwao Sato,$^{3}$
Etsuo Segawa$^{4}$ 
\\ 
{\small
$^{1}$ Mathematics Laboratories, College of Arts and Sciences, Showa University, \\
Fujiyoshida, Yamanashi 403-0005, Japan\\
$^2$ National Laboratory of Scientific Computing - LNCC, \\
Av. Get{\'u}lio Vargas 333, Petr{\'o}polis, RJ, 25651-075, Brazil \\
$^3$ Oyama National College of Technology,\\
Oyama, Tochigi 323-0806, Japan \\
$^4$ Graduate School of Education Center, Yokohama National University,\\
Hodogaya, Yokohama, 240-8501, Japan\\
} 
} 

\date{\empty }
\begin{document}
\maketitle

\begin{abstract}
Staggered quantum walks on graphs are based on the concept of graph tessellation and generalize some well-known discrete-time quantum walk models. In this work, we address the class of 2-tessellable quantum walks with the goal of obtaining an eigenbasis of the evolution operator. By interpreting the evolution operator as a quantum Markov chain on an underlying multigraph, we define the concept of quantum detailed balance, which helps to obtain the eigenbasis. A subset of the eigenvectors is obtained from the eigenvectors of the double discriminant matrix of the quantum Markov chain. To obtain the remaining eigenvectors, we have to use the quantum detailed balance conditions. If the quantum Markov chain has a quantum detailed balance, there is an eigenvector for each fundamental cycle of the underlying multigraph. If the quantum Markov chain does not have a quantum detailed balance, we have to use two fundamental cycles linked by a path in order to find the remaining eigenvectors. We exemplify the process of obtaining the eigenbasis of the evolution operator using the kagome lattice (the line graph of the hexagonal lattice), which has symmetry properties that help in the calculation process.%%%\footnote[0]{{\it Keywords:} Quantum walks; Staggered Model; Quantum Detailed Balance; Graph Tessellation; Spectral analysis}
\end{abstract}

%%%%%%%%%%%
%%%%%%%%%%%
\section{Introduction}
 
The interest of the scientific community on quantum walks has been increasing unremittingly since the first papers of quantum walks on graphs, such as~\cite{FG98,AAKV01}. This interest seems to be based on at least three reasons: (1) the quantum walk is useful to simulate complex physical systems~\cite{KRBD10,CBS10}, (2) it is an important tool to build new quantum algorithms~\cite{SKW03,Amb04,Sze04a}, and (3) it can be implemented directly in laboratories independently of quantum computers~\cite{GASWWMA13,CMQKSPDSSBM15}. Besides those physical and computational aspects, the mathematical aspects of the quantum walk are very rich and have been the focus of many papers~\cite{TS12,HKSS14,Suz16}. 

A quantum walk is defined on a discrete space, which is modeled by a graph. On the other hand, its time evolution can be either continuous or discrete. There are many attempts to prove the equivalence between the continuous-time and discrete-time approaches,  which are successful only at asymptotic limits or on restricted settings~\cite{Str06,Chi10,PP17}. The continuous-time version comes basically in one form, whose evolution operator is local and is obtained from the graph's adjacency or Laplacian matrix. In this case, the spectral analysis of the graph Laplacian~\cite{HS04} helps to understand the quantum dynamics.  The discrete-time versions have evolution operators that are the product of at least two local operators, and the most general models are (1) the coined model~\cite{AAKV01}, (2) Szegedy's model~\cite{Sze04a}, and (3) the staggered model~\cite{PSFG16}. An extensive comparison of those models is performed in~\cite{KPSS18}. 
In the discrete-time case, there is no relation between the eigenvectors of the graph Laplacian and the eigenvectors of the evolution operator.

In this work, we focus on 2-tessellable quantum walks defined on a set of graphs that can be characterized in many ways, for instance, (1) the set of graphs whose clique graphs are 2-colorable, or (2) the set of graphs that are line graphs of bipartite multigraphs. To obtain the evolution operator of a 2-tessellable quantum walk, we need to find two tessellations $\mathcal{T}_1$ and $\mathcal{T}_2$ whose union covers the edges of the graph. A tessellation is a partition of the vertex set into cliques, called polygons or tiles. For instance, $\mathcal{T}_1=\{\alpha_1,\alpha_2,\dots ,\alpha_m \}$ is a tessellation of a graph $G=(V,E)$ if each $\alpha_{\ell}$ is a clique, $\alpha_{\ell}\cap\alpha_j=\emptyset$ when $\ell\neq j$, and $\cup_{\ell=1}^m \alpha_1=V$. $G$ is 2-tessellable if $E(\mathcal{T}_1\cup\mathcal{T}_2)=E(G)$. 
After finding the tessellations, we define two subspaces $\mathcal{A}$ and $\mathcal{B}$ spanned by the polygons of tessellations $\alpha$ and $\beta$, respectively. Using orthogonal projections on these subspaces, there is a standard procedure to obtain self-adjoint unitary operators $H_1$ and $H_2$ associated with the tessellations $\mathcal{T}_1$ and $\mathcal{T}_2$~\cite{PSFG16}. The evolution operator of the quantum walk is $U_\theta=-e^{i\theta_2 H_2}e^{i\theta_1 H_1}$, where $\theta_1$, $\theta_2$ are angles and $i=\sqrt{-1}$.  In this work, we address the case $\theta_1=\theta_2=\theta$~\cite{POM17}.

In order to find the eigenvalues and eigenvectors of the evolution operator of a 2-tessellable quantum walk, we interpret $U_\theta$ as a quantum Markov chain~\cite{Gudder} on the edges of an underlying multigraph $G_{\textrm{un}}$, whose line graph is the original graph, that is, $G=L(G_{\textrm{un}})$. We define the notion of quantum detailed balance using the amplitudes of the polygons and an $(+1)$-eigenvector of a matrix $T$, whose biadjacent matrix is the discriminant of the polygons of tessellations $\mathcal{T}_1$ and $\mathcal{T}_2$. $T$ is a double discriminant matrix and is a self-adjoint operator. We say that a quantum Markov chain is reversible when $T$ has an $(+1)$-eigenvector, which is called a reversible eigenfunction of $T$. A classical Markov chain is obtained from the quantum chain and the classical detailed balance conditions are obtained from the square modulus of the quantum detailed balance conditions. 
When $T$ has a quantum detailed balance, the reversible eigenfunction is the only $(+1)$-eigenvector of $T$, which means that in the quantum case, the invariant state is always a reversible eigenstate and vice versa,  different from the classical Markov chain, whose transition matrix may have a stationary probability distribution even in the irreducible case. 

In the staggered model, the Hilbert space is spanned by the vertex set. We split the Hilbert space as a direct sum of $(\mathcal{A}+\mathcal{B})$ and $(\mathcal{A}+\mathcal{B})^\perp$. The eigenvectors of $U_\theta$ in $(\mathcal{A}+\mathcal{B})$ are inherited from the eigenvectors of $T$. On the other hand, the eigenvectors of $U_\theta$ in $(\mathcal{A}+\mathcal{B})^\perp$ are obtained from the fundamental cycles of the underlying multigraph $G_{\textrm{un}}$. The definition of fundamental cycles relies on the concept of a spanning tree, which is a subgraph of $G_{\textrm{un}}$ that is a tree and includes all vertices of $G_{\textrm{un}}$~\cite{GY05}. Adding one edge to the spanning tree creates a cycle, which is called fundamental cycle. The number of fundamental cycles is equal to the number of edges of $G_{\textrm{un}}$ not in the spanning tree. If $T$ has a quantum detailed balance, there is an eigenvector of $U_\theta$ with support on each fundamental cycle, whose expression is described in this work. The dimension of $(\mathcal{A}+\mathcal{B})^\perp$ is the first Betti number in this case. If $T$ does not have a quantum detailed balance, we have to use two fundamental cycles $c_0$ and $c_1$ in order to obtain an eigenvector. If  $c_0$ and $c_1$ do not have common vertices, we have to link them with a path, and the support of the eigenvector is the cycle-path subgraph.  The dimension of $(\mathcal{A}+\mathcal{B})^\perp$ in this case is the first Betti number minus 1. The eigenvectors in  $(\mathcal{A}+\mathcal{B})^\perp$ play an important role in determining the efficiency of search algorithms on finite graphs.

We use the kagome lattice~\cite{Syo51,Sun13,KoShSu} as an example to show the techniques created in this work because this lattice has interesting symmetries. The dynamic of the staggered quantum walk on the kagome lattice, which is an infinite graph, reduces to a 2-tessellable staggered walk on a triangle, which is the quotient graph of the kagome lattice.
% Since there are tree vertices in the intersection of these polygons, the underlying graph of the reduced graph have two vertices with three multiedges. 
The quantum Markov chain is defined on the underlying graph of the quotient graph and is nonreversible. We show how to use the method based on the fundamental cycles to find the eigenvectors in $(\mathcal{A}+\mathcal{B})^\perp$.

This work generalizes Ref.~\cite{Sze04a}, which introduced the notion of Markov chain-based quantum walks, in many aspects: (1)~We address forms of quantum walks that were not addressed in~\cite{Sze04a}. In fact, Szegedy's model is a subset of quantum walks obtained from the set of 2-tessellable quantum walks if we set $\theta=\pi/2$ and employ a graph whose underlying graph does not have multiedges, (2) Ref.~\cite{Sze04a} does not describe the eigenvectors associated with the cycle-path space $(\mathcal{A}+\mathcal{B})^\perp$, and (3) Ref.~\cite{Sze04a} does not describe the quantum detailed balance conditions, which have been proposed for quantum walks in this work, as far as we know.
This work generalizes Ref.~\cite{KIS18} in two directions: (1)~We consider the staggered model with Hamiltonians with a generic $\theta$ while Ref.~\cite{KIS18} addressed only the case $\theta=\pi/2$, and (2)~we obtain a complete eigenbasis while Ref.~\cite{KIS18} missed a subset of eingenvectors corresponding to the space $(\mathcal{A}+\mathcal{B})^\perp$.
Primitive ideas related to the cycle-path space were first introduced in~\cite{Sh} from the point of view of simple random walks on line and para-line graphs. 
After that the ideas are applied to the spectral analysis of twisted random walks~\cite{HS04} and twisted Grover walks~\cite{HKSS14}. 
The twisted Grover walk on the para-line graph is a special case of the model presented in this work, since there is an equivalence between the Grover walk and the 2-tessellable quantum walk~\cite{KPSS18}. Here, we have further developed the concept of cycle-path space, which can now be used in a more general setting.
The structure of this paper is as follows. In Sec.~\ref{Sec:SQW} we define the staggered quantum walk, its evolution operator, and the subspaces spanned by the polygons.
In Sec.~\ref{eigvals} we obtain the eigenvalues of the evolution operator of 2-tessellable quantum walks.
In Sec.~\ref{eigvecs}  we define the quantum detailed balance conditions and obtain an eigenbasis of the evolution operator $U_\theta$ of 2-tessellable quantum walks. This section is divided into three subsections, which address the eigenvectors in the following subspaces: $(\mathcal{A}\cap \mathcal{B})$, $(\mathcal{A}+\mathcal{B})$, and $(\mathcal{A}+\mathcal{B})^\perp$. In Section~\ref{sec:maintheo} we summarize our results in a theorem. In Section~\ref{kagome} we use the kagome lattice as an example.

%%%%%%%%%%%
%%%%%%%%%%%
\section{The 2-tessellable quantum walk}\label{Sec:SQW}
The evolution operator of a staggered quantum walk on a graph $G=(V,E)$ is associated with a graph tessellation cover. A tessellation cover is a set of graph tessellations whose union covers the edge set. The formal definition is as follows~\cite{PSFG16}.
\begin{definition}
A \emph{graph tessellation} $\mathcal{T}_1$ of $G=(V,E)$ is a partition of the vertex set $V$ into cliques. An edge \emph{belongs} to the tessellation $\mathcal{T}_1$ if and only if the edge endpoints belong to the same clique in ${\mathcal{T}}_1$. 
The set of edges belonging to ${\mathcal{T}}_1$ is denoted by ${\mathcal{E}}({\mathcal{T}}_1)$. An element of the tessellation is called a \emph{polygon} (or \emph{tile}). The \emph{size} of a tessellation ${\mathcal{T}}_1$ is the number of polygons in ${\mathcal{T}}_1$. A \emph{tessellation cover} of size~$k$ of $G$ is a set of $k$~tessellations ${\mathcal{T}}_1,...,{\mathcal{T}}_k$ whose union covers the edges, that is, $\cup_{i=1}^k\,{\mathcal{E}}({\mathcal{T}}_i)=E(G)$. If there is a tessellation cover of size at most~$k$, graph~$G$ is called \emph{$k$-tessellable}. 
\end{definition}

A staggered quantum walk on a $k$-tessellable graph is called $k$-tessellable quantum walk. Since this paper addresses $2$-tessellable quantum walks, we assume from now on that $k=2$. A graph is 2-tessellable if and only if its clique graph is 2-colorable~\cite{Por16b}. Besides, it is known that the clique graph of a graph $G$ is 2-colorable if and only if $G$ is the line graph of a bipartite multigraph~\cite{Pet03}. Then, throughout this paper, $G$ is the line graph of a bipartite multigraph on which we now define a 2-tessellable quantum walk.

Suppose that $\{\mathcal{T}_1,\mathcal{T}_2\}$ is a tessellation cover of $G$, where
$\mathcal{T}_1=\{\alpha_1,\alpha_2,\dots ,\alpha_m \}$, $\mathcal{T}_2=\{\beta_1,\beta_2,\dots,\beta_n\}$, where $m=|\mathcal{T}_1|$ and $n=|\mathcal{T}_n|$. Let us assume that $n \ge m$.
Generic polygons of $\mathcal{T}_1$ and $\mathcal{T}_2$ are denoted by $\alpha_i$ and $\beta_j$, respectively. 
Let $\mathcal{H}$ be the Hilbert space $\ell^2(V)$, that is, $\mathcal{H}$ is spanned by the vertices. 
The standard basis of $\mathcal{H}$ is denoted by $\{|u\rangle \;|\; u\in V\}$ which coincides with the delta function on each vertex. 
We assign a complex-valued unit vector in $\mathcal{H}$ to each polygon: 
$|\alpha_1\rangle$, $|\alpha_2\rangle$, $\dots$, $|\alpha_m\rangle$, $|\beta_1\rangle$, $|\beta_2\rangle$, $\dots$, $|\beta_n\rangle$, that is, if $u\notin \alpha_i$, then $\langle  u|\alpha_i \rangle=0$ $(i=1,\dots,m)$, and 
if $v\notin \beta_j$, then $\langle  v|\beta_j \rangle=0$ $(j=1,\dots,n)$. 
For any $u\in V$, let $\mathcal{T}_1(u)\in \mathcal{T}_1$ and $\mathcal{T}_2(u)\in \mathcal{T}_2$ be the polygons that include $u$. A vertex $u$ belongs to exactly two polygons, which we call $\alpha_i$ and $\beta_j$. 
Let $a, b \in \mathcal{H}$ be the functions such that for $u\in \alpha_i \cap \beta_j$, 
\begin{equation}\label{Eq:a_b}
a(u) := \langle u|\alpha_i\rangle; \;b(u) := \langle u|\beta_j\rangle.
\end{equation}
%Let $a_i$ and $b_j$ be restricted versions of vectors of $a$ and $b$ to $\alpha_i$ and $\beta_j$, repectively, such that 
%	\[ a_i(u)=\begin{cases} a(u), & \text{if $u\in\alpha_i$}\\ 0, & \text{otherwise,}\end{cases}\;\;\;\; 
%           b_j(u)=\begin{cases} b(u), & \text{if $u\in\beta_j$}\\ 0, & \text{otherwise.}\end{cases}\]
Let $A: \ell^2(\mathcal{T}_1)\to \ell^2(V)$ and $B: \ell^2(\mathcal{T}_2)\to \ell^2(V)$
be $A=[|\alpha_1\rangle\;|\alpha_2\rangle\;\cdots\;|\alpha_m\rangle]$ and $B=[|\beta_1\rangle\;|\beta_2\rangle\;\cdots\;|\beta_n\rangle]$, that is, 
$(Af)(u)=a(u)f(\mathcal{T}_1(u))$ and $(Bg)(u)=b(u)g(\mathcal{T}_2(u))$ for any $u\in V$, $f\in \ell^2\left(\mathcal{T}_1\right)$, and $g\in \ell^2\left(\mathcal{T}_2\right)$.
Then, $(A^{\dagger}\psi)(\alpha_i)=\langle \alpha_i| \psi\rangle$ and $(B^{\dagger}\psi)(\beta_j)=\langle \beta_j| \psi\rangle$. 

Hamiltonians $H_A$ and $H_B$ associated with tessellations $\mathcal{T}_1$ and $\mathcal{T}_2$, respectively, are defined by
\begin{align}
H_A &=2AA^{\dagger}-I_{\mathcal{H}}, \label{Eqs:H_A_H_B_1}\\
H_B &=2BB^{\dagger}-I_{\mathcal{H}}. \label{Eqs:H_A_H_B_2}
\end{align}  
Note that $AA^{\dagger}$ and $BB^{\dagger}$ are projection operators on 
	\begin{align} 
        \mathcal{A}: &= \mathrm{span}\{|\alpha_i\rangle \;|\; i=1,2,\dots,{m} \}=\left\{Af \;|\;f\in \ell^2\left(\mathcal{T}_1\right)\right\},\label{Eq:projectionA} \\
	\mathcal{B}: &= \mathrm{span}\{|\beta_j\rangle \;|\; j=1,2,\dots,{n} \}=\left\{Bg \;|\;g\in \ell^2\left(\mathcal{T}_2\right)\right\},\label{Eq:projectionB}
        \end{align} 
respectively, and besides
	\begin{equation}\label{Eq:identity} 
        A^{\dagger}A=I_{\ell^2(\mathcal{T}_1)},\;B^{\dagger}B=I_{\ell^2(\mathcal{T}_2)}.
        \end{equation} 
Then, $H_A$ and $H_B$ are self-adjoint unitary operators, that is, $H_A^2=H_B^2=I_{\mathcal{H}}$. 
The evolution operator of a 2-tessellable quantum walk is defined as~\cite{POM17}
\begin{equation}\label{Eq:U_theta}
	U_\theta=-e^{i\theta H_B}e^{i\theta H_A},
\end{equation}	
where $\theta\in (0,\pi)$ is a fixed parameter and a minus sign was added for convenience.
The values $\theta=0$ and $\pi$ are excluded because the walk is trivial in these cases. In the next sections, we address the problem of finding the eigenvalues and eigenvectors of $U_\theta$.

%%%%%%%%%%%
%%%%%%%%%%%
\section{Eigenvalues}\label{eigvals}
%%%%%%%%%%%
%%%%%%%%%%%

Let $U_{\theta}$ be the evolution operator of a 2-tessellable staggered quantum walk on $G$, as described in Section~\ref{Sec:SQW}. 
In this section we obtain the spectrum of $U_\theta$. Let us start with a useful lemma.
\begin{lemma}\label{Lemma_1}
Let
\[ M=I_{{\cal H}}+\frac{b}{a}BB^\dagger , \  N=I_{{\cal H}}-\frac{b}{a+b}BB^\dagger, \]
where $a,b$ are complex numbers. 
Then, $MN=I$ and
\[\det(M)=\left(\frac{a+b}{a}\right)^n.\]
\end{lemma}
Besides Lemma~\ref{Lemma_1}, the proof of the next theorem uses the fact
\[
\det ( {I}_{m_1} - {M_1} {M_2} )= \det ( {I}_{m_2} - {M_2} {M_1} ),
\] 
where $M_1$ and $M_2$ are $m_1\times m_2$ and $m_2\times m_1$ matrices, respectively.
\begin{theorem}\label{theo:minpoly}
The characteristic polynomial of $U_\theta$ is
\[
\det(\lambda I_{{\cal H}} - U_\theta)
=(\lambda+1)^{n-m}\left(\lambda+e^{-2i\theta}\right)^{\nu-m-n}\det\left((\lambda+1)^2 I_m-4\lambda \sin^2(\theta) A^\dagger B B^\dagger A\right),
\]
where $\nu=|V|$.
\end{theorem}
\begin{proof}
Using equation~(\ref{Eq:U_theta}), we have
\[  
\det ( I_{{\cal H}} -u \ U_{\theta } ) = \det ( I_{{\cal H}} +u  e^{i \theta H_B} e^{i \theta H_A} ), 
\]
where $|u|=1$. Commuting the order of the matrices inside the determinant, using that $ e^{i \theta H_A}=\cos \theta I_{{\cal H}} +i \sin \theta H_A$, $ e^{i \theta H_B}=\cos \theta I_{{\cal H}} +i \sin \theta H_B$, equations~(\ref{Eqs:H_A_H_B_1}) and~(\ref{Eqs:H_A_H_B_2}), we obtain
\begin{align*}
\det ( I_{{\cal H}} -u \ U_{\theta } ) = \det \Big((1+u e^{-2i \theta } ) I_{{\cal H}} +2iu \sin(\theta) e^{-i \theta } BB^\dagger - \\
2u \sin^2 (\theta) AA^\dagger \big(2BB^\dagger-(1+i \cot \theta ) I_{{\cal H}} \big)\Big).
\end{align*}
Factoring out the term $(1+u e^{-2i \theta } )$ and using that the determinant of a matrix product of square matrices equals the product of their determinants, we obtain
\begin{align*}
\det ( I_{{\cal H}} -u \ U_{\theta } ) & = a^{\nu } \det \Big( I_{{\cal H}} - \frac{c}{a} AA^\dagger \big(2BB^\dagger-(1+i \cot \theta ) I_{{\cal H}} \big)
\big(I_{{\cal H}} + \frac{b}{a} BB^\dagger\big)^{-1} \Big)\times \\
& \det \Big( I_{{\cal H}} + \frac{b}{a} BB^\dagger\Big) ,
\end{align*}
where $a=1+u e^{-2i \theta }$, $b=2iu \sin \theta e^{-i \theta }$, and $c=2u \sin {}^2 \theta$.

By Lemma~\ref{Lemma_1}, we have 
\[
\det \Big( {I}_{{\cal H}} + \frac{b}{a} BB^\dagger \Big)= \frac{(a+b)^n }{a^n } . 
\]
and 
\[
\Big( {I}_{{\cal H}} - \frac{b}{a} BB^\dagger \Big)^{-1} = {I}_{{\cal H}} - \frac{b}{a+b} BB^\dagger . \\ 
\]
Using these identities and equation (\ref{Eq:identity}), we obtain 
\begin{align*}
\det ( I_{{\cal H}} -u \ U_{\theta } ) &  = a^{ \nu -n} (a+b)^n\times \\
& \det \Big( I_{{\cal H}} - \frac{c}{a} AA^\dagger 
\Big(\frac{2a+b+ib \cot \theta }{a+b} BB^\dagger  -  
(1+i \cot \theta ) I_{{\cal H}} \Big)\Big).
\end{align*}
Commuting the order of the matrices inside the determinant and factoring out the denominator, we obtain
\begin{align*}
\det ( I_{{\cal H}} -u \ U_{\theta } ) & = a^{ \nu -m-n} (a+b)^{n-m}\times \\
& \det \Big((a+c+ic \cot \theta )(a+b) I_m -  
c(2a+b+ib \cot \theta ) A^\dagger BB^\dagger A \Big) . 
\end{align*}
Using that $a+b= a+c+ic \cot \theta = 1+u$ and $2a+b+ib \cot \theta = 1$, we obtain
\begin{align*}
\det ( I_{{\cal H}} -u \ U_{\theta } ) & =(1+u e^{-2i \theta } )^{ \nu -m-n} (1+u)^{n-m}\times \\
& \det \Big((1+u)^2 I_{m} - 4u \sin^2(\theta) A^\dagger BB^\dagger A\Big) .
\end{align*}
Setting $u=1/ \lambda$ in the above equation, we obtain the characteristic polynomial of $U_\theta$ in terms of the characteristic polynomial of $A^\dagger BB^\dagger A$.
\end{proof}

The spectrum of $U_{\theta}$ is obtained from the solutions $\lambda$ of the equation $\det(\lambda I_{{\cal H}} - U_\theta)=0$. We need the next lemma before describing the spectrum of $U_{\theta}$.
\begin{lemma}\label{lem:bipartite}
Let $T$ be the following $(n+m)\times (n+m)$ matrix: 
	\[  T=\begin{bmatrix} 0 & A^{\dagger} B \\  B^{\dagger} A & 0 \end{bmatrix}. \]
Then,  
\[\det(\mu I-T)=\mu^{n-m}\det(\mu^2 I-A^{\dagger}BB^{\dagger}A).\] 
\end{lemma}

The next corollary describes the spectrum of $U_{\theta}$ in terms of the spectrum of $T$.
 
\begin{corollary}\label{cor:eigenvalues} 
The spectrum of $U_\theta$ is
	\[ \sigma(U_\theta)=\left\{e^{2i\phi} \;\bigg|\; \left(e^{2i\phi}+e^{-2i\theta}\right)^{\nu-m-n} \det\left(\frac{\cos\phi}{\sin\theta}  I_{n+m}-T \right)=0\right\}. \]

\end{corollary}
\begin{proof}
Using Theorem~\ref{theo:minpoly} and factoring out $4\lambda \sin^2(\theta)$, the spectrum of $U_{\theta}$ is obtained from equation
\begin{align*}
(\lambda+1)^{n-m}\left(\lambda+e^{-2i\theta}\right)^{\nu-m-n}\det\left(\frac{(\lambda+1)^2}{4\lambda \sin^2(\theta)} I_m- A^\dagger B B^\dagger A\right)=0.
\end{align*}
Setting $\lambda=e^{2 i\phi}$ and $\cos\phi/\sin\theta=\mu$, using the fact that $(\lambda+1)^2/(4\lambda)=\cos^2\phi$, we obtain
\begin{align*}
\left(e^{2i\phi}+e^{-2i\theta}\right)^{\nu-m-n}\mu^{n-m}\det(\mu^2 I-A^{\dagger}BB^{\dagger}A)=0.
\end{align*}
Using Lemma~\ref{lem:bipartite}, we obtain $\sigma(U_\theta)$. 
\end{proof}

Lemma~\ref{lem:bipartite} implies that $\dim(\ker(T))\geq n-m$ and, since $\cos\phi=\mu\sin\theta$, the 0-eigenvalues of $T$ are associated with the $(-1)$-eigenvalues of $U_\theta$. Besides, the $(+1)$-eigenvalues of $T$ are associated with the $(-e^{-2i \theta })$-eigenvalues of $U_\theta$.  Summarizing, the spectrum of $U_\theta$ can be described as follows (see Fig.~\ref{fig:mapping}): 
\begin{enumerate} 
\item \label{eq:062712} $\lambda=-e^{-2i \theta }$ with multiplicity at least $\max\{\nu-m-n,0\}$. 
\item \label{eq:062711} There are $\max\{n+m,\nu\}$ eigenvalues $\lambda = e^{2i\phi}$, where $\cos\phi= \mu\sin\theta$ and
\[ \mu\in \begin{cases} \sigma(T) & \text{if $n+m\leq \nu$,}\\ 
\sigma(T)\setminus\{1\} & \text{if $n+m> \nu$,} \end{cases} \]
and, in particular, $\lambda=-1$ with multiplicity at least $n-m$.
\end{enumerate} 
The multiplicity of $(-e^{-2i\theta})$ depends on the reversibility of $T$ (see Theorem~\ref{thm:main}). 
Note that $\nu-m-n<0$ if and only if $\nu=m+n-1$. Take 
for example $G=P_3$ ($P_3=\bullet$-----$\bullet$-----$\bullet$) and $V(G)=\{1,2,3\}$ with $\mathcal{T}_1=\{\{1,2\},\{3\}\}$, $\mathcal{T}_2=\{\{1\},\{2,3\}\}$, which has $\nu=3$ and $m=n=2$.

\begin{figure}
\begin{center}
	\includegraphics[width=90mm]{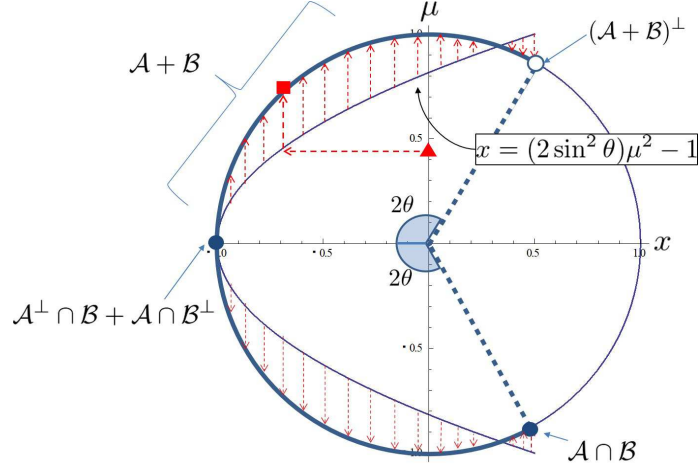}
\end{center}
\caption{The mapping from $\mu\in \sigma(T)$ to $2\phi\in \sigma(U|_{\mathcal{A+B}})$: 
The red triangle {\small\textcolor{red}{$\blacktriangle$}} on the $\mu$-axis is mapped to the red square {\scriptsize\textcolor{red}{$\blacksquare$}} on the unit circle in the complex plane. 
The blue curve depicted in bold face on the unit circle corresponds to $\mathcal{A}+\mathcal{B}$. 
The empty point at angle $(\pi-2\theta)$ corresponds to $(\mathcal{A}+\mathcal{B})^\perp$ 
and the full point at angle $(\pi+2\theta)$ corresponds to $\mathcal{A}\cap \mathcal{B}$. 
The point at angle $\pi$ mapped from $0\in \sigma(T)$ corresponds to $\mathcal{A}^\perp\cap \mathcal{B}+\mathcal{A}\cap \mathcal{B}^\perp$. 
}
\label{fig:mapping}
\end{figure}

%%%%%%%%%%%
%%%%%%%%%%%
\section{Eigenvectors}\label{eigvecs}
%%%%%%%%%%%
%%%%%%%%%%%

%
In the previous section, we have shown that the self-adjoint matrix $T$ plays a key role in the description of the spectrum of $U_\theta$. 
In this section, we discuss some extra properties of $T$ in order to obtain the eigenvectors of $U_\theta$. 
Let $\mathcal{K}$ be the Hilbert spanned by $\mathcal{T}_1$ and $\mathcal{T}_2$, that is, 
$\mathcal{K}:=\{ \psi: \mathcal{T}_1 \sqcup \mathcal{T}_2\to \mathbb{C} \;|\; ||\psi||^2<\infty\}$. 
We employ the standard inner product.

As in Lemma~\ref{lem:bipartite}, we write $T: \mathcal{K}\to \mathcal{K}$ as
	\[ T = \begin{bmatrix} 0 & T_{AB} \\ T_{BA} & 0 \end{bmatrix}, \]
where $T_{AB}=A^{\dagger}B$ and $T_{BA}=T_{AB}^{\dagger}$. 
The entries of $T$ are given by $(T_{AB})_{i,j}=\langle \alpha_i|\beta_j \rangle$. 
Now we define the notions of reversible eigenfunction and quantum detailed balance (QDB) for a pair $(a,b)$, where $a$ and $b$ are given by Eq.~(\ref{Eq:a_b}). 
\begin{definition}\label{def:DBC}
The pair $(a,b)$ obeys the quantum detailed balance conditions if there exists an eigenfunction $\pi$ of $T$ such that 
	\[ a(u)\pi(\mathcal{T}_1(u))=b(u)\pi(\mathcal{T}_2(u)),\]
for every $u\in V(G)$. 
We call this function $\pi$ a reversible eigenfunction.
\end{definition}

We say that $T$ is reversible or $T$ has a quantum detailed balance if there is a pair $(a,b)$ that obeys the QDB conditions.
A useful property of the spectrum of $T$ is as follows. 
\begin{lemma} The spectrum of $T$ obeys
$\sigma(T)\subseteq [-1,1]$. 
\end{lemma}
\begin{proof}
Let $f\oplus g\in \ker(\lambda I-T)$. 
Then,
	\begin{align*}
         |\lambda^2|^2\, ||f||^2 
         	&= ||T_{AB}T_{BA}f||^2\leq\langle B^{\dagger}Af, B^{\dagger}Af \rangle \\
        	&= \langle Af, BB^{\dagger}Af \rangle \\
                &\leq \langle Af, Af \rangle \\
                &= ||f||^2.
	\end{align*}
Since $T$ is self-adjoint and $\lambda^2\le 1$, the result follows.
\end{proof}

Define an underlying bipartite multigraph $G_{\textrm{un}}=(V_{\textrm{un}},E_{\textrm{un}})$ whose vertex set is $V_{\textrm{un}} = \mathcal{T}_1 \sqcup \mathcal{T}_2$ and two vertices are adjacent if and only if $|\alpha \cap \beta|>0$ for $\alpha\in \mathcal{T}_1$ and $\beta\in \mathcal{T}_2$ and the number of multiple edges is given by $|\alpha \cap \beta|$. The adjacency matrix of $G_{\textrm{un}}$ is obtained in the following way. Let $A'$ and $B'$ be the matrices obtained from $A$ and $B$ by replacing the nonzero entries by 1, respectively. The adjacency matrix is 
\[T'=\begin{bmatrix} 0 & (A')^TB' \\ (B')^TA' & 0 \end{bmatrix}.\] 
The entries of $T'$ are nonnegative integers and $T'_{ij}$ is the number of multiedges linking vertices $\alpha_i$ and $\beta_j$ of $G_{\textrm{un}}$.  Note that (1)~$G_{\textrm{un}}$ is an intersection multigraph whose family of sets are the polygons of the tessellations $\mathcal{T}_1$ and $\mathcal{T}_2$, and (2)~$G_{\textrm{un}}$ is a root multigraph of $G$, that is, the line graph of $G_{\textrm{un}}$ is $G$ and there is a one-to-one correspondence between $E_{\textrm{un}}$ and $V(G)$.

In order to find the eigenvectors of $U_\theta$, we decompose the total state space as $\mathcal{H}=(\mathcal{A}+\mathcal{B}) \oplus (\mathcal{A}+\mathcal{B})^\perp$. In the next subsection, we address the subspace $\mathcal{A} \cap \mathcal{B}\subset \mathcal{A}+\mathcal{B}$, which is the one most amenable in terms of algebraic manipulations. In Subsection~\ref{subsec:Inherited}, we obtain the eigenvectors in subspace $(\mathcal{A}+\mathcal{B})$, and in Subsection~\ref{subsec:Cycle-path} obtain the eigenvectors in subspace $(\mathcal{A}+\mathcal{B})^\perp=\mathcal{A}^\perp \cap \mathcal{B}^{\perp}$.

%%%%%%%%%%%%%%
\subsection{Space $\mathcal{A} \cap \mathcal{B}$}\label{subsec:AcapB}

The following lemma shows a useful necessary and sufficient condition that the spectrum of $T$ must obey in order to include eigenvalues $\pm 1$, which is important for obtaining the eigenvectors of $U_\theta$. Besides, in the proof of this lemma, we describe a classical Markov chain induced by the quantum chain and how to obtain the classical detailed balance conditions from the quantum detailed balance conditions.
\begin{lemma}\label{lem:DBC} 
$T$ is reversible with a reversible eigenfunction $\pi=\pi_1\oplus \pi_2$
if and only if $\dim(\ker(I-T))= \dim(\ker(I+T))= 1$, $\ker(I-T)=\mathbb{C}\,(\pi_1\oplus \pi_2)$, and $\ker(I+T)=\mathbb{C}\,(\pi_1\oplus(-\pi_2))$.  
\end{lemma}
\begin{proof}
The following equivalences hold
\begin{align*}
f=f_1\oplus f_2\in \ker(I-T) 
	& \Leftrightarrow T_{BA}f_1=f_2,\;T_{AB}f_2=f_1 \\
	& \Leftrightarrow Af_1=Bf_2 \\
        & \Leftrightarrow a(u)f_1(\mathcal{T}_1(u))=b(u)f_2(\mathcal{T}_2(u)), 
\end{align*}
for any $u\in V$, which means that $(a,b)$ obeys the QDB conditions and $f$ is a reversible measure. 
The second equivalence is obtained as follows. 
If $T_{BA}f_1=f_2,\;T_{AB}f_2=f_1$, then $B^{\dagger}(Af_1-Bf_2)=0,\;A^{\dagger}(Af_1-Bf_2)=0$, which implies
$Af_1=Bf_2$; the opposite direction is obtained by taking $B^{\dagger}$ and $A^{\dagger}$ to both sides. 
Then, 
\[\ker(I-T)=\{\pi \;|\; \pi \mathrm{\;is\;a\;reversible\;measure}\}. \]
Let us show that $\dim(\ker(I-T))\leq 1$.
Note that if $T$ is nonreversible, then $\dim(\ker(I-T))=0$. 
Now we consider the reversible case and show $\dim(\ker(I-T))=1$. 

By taking the square modulus of both sides of the QDB equation, and putting $p(e):=|a(e)|^2$ and $q(e):=|b(e)|^2$, 
we have $p(e)\zeta(\mathcal{T}_1(e))=q(e)\zeta(\mathcal{T}_2(e))$ for every $e\in E(G_{\textrm{un}})\simeq V(G)$, where $\zeta(\gamma)=|\pi(\gamma)|^2$ for $\gamma\in V(G_{\textrm{un}})=\mathcal{T}_1 \sqcup \mathcal{T}_2$. Now we consider a classical Markov chain on $G_{\textrm{un}}$ with the stochastic transition matrix $P$ such that 
the transition probability from $\alpha\in \mathcal{T}_1$ to $\beta\in \mathcal{T}_2$ is 
\[ \langle\delta_\beta, P\delta_\alpha\rangle=\sum_{e:\mathcal{T}_1(e)=\alpha,\;\mathcal{T}_2(e)=\beta}p(e), \]
the transition probability from $\beta\in \mathcal{T}_2$ to $\alpha\in \mathcal{T}_1$ is 
\[ \langle\delta_\alpha,P\delta_\beta\rangle=\sum_{e:\mathcal{T}_1(e)=\alpha,\;\mathcal{T}_2(e)=\beta}q(e). \]
Note that the classical detailed balance conditions are
\[\langle\delta_\beta, P\delta_\alpha\rangle\zeta(\alpha)=\langle\delta_\alpha,P\delta_\beta\rangle\zeta(\beta).  \] 
By the Perron-Frobenius theorem, we have $\ker(I-P)=\mathbb{C}\,\zeta$. 
Besides, $T=\mathcal{M}^{-1}P\mathcal{M}$, where $(\mathcal{M}f)(\gamma)=\bar{\pi}(\gamma)f(\gamma)$, $\forall f\in \mathcal{K}$, and $\forall\gamma\in \mathcal{T}_1\sqcup \mathcal{T}_2$.  
If there are two reversible eigenfunctions $\pi\neq \pi'$ of $T$, then $\pi'(\gamma)=e^{i \eta_\gamma}\pi(\gamma)$ for some $\eta_{\gamma}\in \mathbb{R}$. 
Then, by the definition of the QDB, we conclude that $\pi'=e^{i\eta_*}\pi$ for some constant $\eta_*\in \mathbb{C}$, and we obtain $\dim(\ker(I-T))=1$. 

In general, by the property of the bipartiteness, $f_1\oplus f_2\in \ker(\mu-T)$ if and only if $f_1\oplus (-f_2)\in \ker(\mu+T)$. 
We conclude that if $T$ is reversible and $\pi_1\oplus \pi_2$ is a reversible eigenfunction, then $\ker(I+T)=\mathbb{C}\,(\pi_1\oplus(-\pi_2))$.
\end{proof}

\begin{corollary}\label{lem:DBC_corol}
$T$ is nonreversible if and only if $\dim(\ker(I-T))=\dim(\ker(I+T))=0$.
\end{corollary}
Recall that $\mathcal{A}$ and $\mathcal{B}$ are defined in Eqs.~(\ref{Eq:projectionA}) and~(\ref{Eq:projectionB}). If $\psi\in \mathcal{A}\cap \mathcal{B}$, then there exist $f\in \ell^2(\mathcal{T}_1)$ and $g\in\ell^2(\mathcal{T}_2)$ such that $\psi=Af=Bg$. This implies that $a(u)f(\mathcal{T}_1(u))=b(u)g(\mathcal{T}_2(u))$ for every $u\in V$. 
By putting 
	\[ \pi(\gamma):=\begin{cases} f(\gamma) & \text{: $\gamma\in \mathcal{T}_1$,} \\ g(\gamma) & \text{: $\gamma\in \mathcal{T}_2$,} \end{cases} \]
then $a(u)\pi(\mathcal{T}_1(u))=b(u)\pi(\mathcal{T}_2(u))$ for every $u\in V$ and $(a,b)$ obeys the QDB conditions and $\pi$ is a reversible eigenfunction. If $f\oplus g$ is a reversible eigenfunction, then $f\oplus (-g)\in \ker(I+T)$. 
\begin{lemma}\label{lem:dimAB}
The dimension of $\mathcal{A}\cap \mathcal{B}$ obeys
	\[ \dim(\mathcal{A}\cap \mathcal{B})\leq 1. \]
Moreover, $\dim (\mathcal{A}\cap \mathcal{B})=1$ if and only if $T$ is reversible. 
The subspace $\mathcal{A}\cap \mathcal{B}$ is invariant under the action of $U_\theta$ whose eigenvalue is $-e^{2i\theta}$, 
and the eigenspace is described by 
\begin{equation}\mathcal{A}\cap \mathcal{B}=\mathbb{C}A\pi_1=\mathbb{C}B\pi_2, \end{equation}
where $\pi_1\oplus \pi_2$ is a reversible eigenfunction. 
\end{lemma}
\begin{proof}
Let us prove the nontrivial part. If $\psi\in \mathcal{A}\cap \mathcal{B}$, then $\psi$ is a $(e^{i\theta})$-eigenvector of $e^{i\theta H_A}$ and $e^{i\theta H_B}$ because $H_A$ and $H_B$ are unitary and self-adjoint operators. Then, $\psi$ is an eigenvector of $U_\theta$ with eigenvalue $(-e^{2i\theta})$.
\end{proof}
These lemmas and corollary show that the reversibility of $T$ plays an important role in the spectral analysis. 

%%%%%%%%%%%%%%
\subsection{Space $\mathcal{A}+\mathcal{B}$ inherited from $T$}\label{subsec:Inherited}
%%%%%%%%%%%%%%

The next lemma shows that the action of $U_\theta$ on $\mathcal{A}+\mathcal{B}$ can be expressed in terms of $T_{AB}$ and $T_{BA}$. 
\begin{lemma}\label{lem:masterEq}
Let $L: \ell^2(\mathcal{T}_1)\oplus \ell^2(\mathcal{T}_2)\to \ell^2(V)$ be defined as $L(f\oplus g)=Af+Bg$ or, equivalently, $L=[A\;B]$. Then, we have 
	\begin{equation}
        U_\theta L=L\Lambda_\theta,
        \end{equation}
where
	\begin{equation}
        \Lambda_\theta=-\begin{bmatrix} I & 2ie^{-i\theta}\sin(\theta) T_{AB} \\ 2ie^{i\theta} \sin(\theta) T_{BA} & I-4\sin^2(\theta) T_{BA}T_{AB} \end{bmatrix}.
        \end{equation}
\end{lemma}
\begin{proof}
The proof is obtained by employing Eqs.~(\ref{Eq:U_theta}), property~(\ref{Eq:identity}), the definitions of $T_{AB}$ and $T_{BA}$, and by performing a straightforward calculation.
\end{proof}
If $f\oplus g\in\ker(\lambda-\Lambda_\theta)\setminus \ker L$, then Lemma~\ref{lem:masterEq} implies that $Af+Bg\in \ker(\lambda-U_\theta)\setminus\{0\}$. This shows that the spectral decomposition of $\Lambda_\theta$ helps to obtain the spectral decomposition of $U_\theta$. In the next subsection, we focus on the spectral decomposition of $\Lambda_\theta$, and in the following one we address the kernel of $L$.
%%%After that we eliminate an extra part from the above result. 
%%%%
\subsubsection{Spectral decomposition of $\Lambda_\theta$} 
%%%%
In Corollary~\ref{cor:eigenvalues}, we have defined $\phi$ using equation $\cos\phi=\mu\sin\theta$, where $\mu$ is an eigenvalue of $T$ and $\phi\in [0,\pi)$. Sometimes we denote $\phi$ by $\phi(\mu)$ to stress its relation with $\mu$. The eigenspace of $\Lambda_\theta$ associated with the eigenvalue $e^{2i\phi(\mu)}$ is related with the eigenspace of $T$ associated with eigenvalue $\mu$ as described by the following lemma. 
\begin{lemma}\label{lem:Kokubuncho}
Assume that $\theta\notin\{0,\pi\}$. Then, 
	\begin{equation}
        \ker(e^{2i\phi(\mu)}-\Lambda_\theta)=D\ker(\mu -T),
        \end{equation}
where $D$ is the diagonal matrix defined by
	\[D(f\oplus g)=f\oplus \left(ie^{i(\theta+\phi(\mu))}g\right). \]
\end{lemma}
\begin{proof}
After applying the Gaussian elimination method to $(e^{2i\phi}-\Lambda_\theta)$, we obtain  
	\[ \ker(e^{2i\phi}-\Lambda_\theta)=\ker \left(\frac{\cos\phi}{\sin\theta}-DTD^{-1}\right). \]
%On the other hand, if $f\oplus g\in \ker(\mu-T)$, then $D(f\oplus g)\in \ker(\mu-DTD^{-1})$. 
As a last step we use the fact that $\ker(\mu-DTD^{-1})=D\ker(\mu-T)$. 
\end{proof}
The last lemma shows that the spectrum of $\Lambda_\theta$ can be obtained from the spectrum of $T$, that is, $e^{2i\phi}$ can be obtained from $\mu$ using $\cos\phi=\mu\sin\theta$. 
Now we list some relevant observations about the eigenvalues of $\Lambda_\theta$ (see Fig.~\ref{fig:mapping}). 
\begin{enumerate}
\item There is a spectral gap if $\theta\neq \pi/2$. In fact,  $\sigma(\Lambda_\theta)\subset\{e^{2i\phi} \;|\; \cos 2\phi\in [-1,\cos(\pi-2\theta)]\}$ because $\cos 2\phi=2\mu^2\sin^2\theta-1$ and $|\mu|\leq 1$. 

\item Map $\phi$ is a bijection if $\theta\neq \pi/2$. In fact, all eigenvalues of $\Lambda_\theta$ are inherited from eigenvalues of $T$ and there is a one-to-one correspondence between $\sigma(\Lambda_\theta)$ and $\sigma(T)$. 
On the other hand, if $\theta=\pi/2$, the spectral gap disappears. In this case,  if $T$ has a quantum detailed balance, $\phi$ is not a bijection because $\phi(1)=\phi(-1)=0$. 

\item The spectrum of $\Lambda_\theta$ is symmetric. In fact, $e^{2i\phi}\in \sigma(\Lambda_\theta)$ if and only if $e^{2i(\pi-\phi)}=e^{-2i\phi}\in \sigma(\Lambda_\theta)$ 
because there is an equivalence between ``$\mu\in \sigma(T)$ with $f\oplus g\in \ker(\mu-T)$" 
and ``$-\mu\in \sigma(T)$ with $f\oplus (-g)\in \ker(\mu+T)$" since $G_\textrm{un}$ is bipartite. 
\end{enumerate}
%
%%%%
\subsubsection{Kernel of  $L$}
From now on, to obtain an eigenfunction of $U_\theta$ restricted to the subspace $\mathcal{A}+\mathcal{B}$, 
we use a lift-up operation $LD$ from the set of the eigenfunctions of $\Lambda_\theta$ in $\ell^2(\mathcal{T}_1)\oplus\ell^2(\mathcal{T}_2)$ to the original space $\ell^2(V)$. 
Recall that the eigenfunctions of $\Lambda_\theta$ should not be in the kernel of $L$. 
It is therefore natural to characterize the kernel of $L$. 
\begin{lemma}\label{lem:kerL}
 \[ \ker(L)=\ker(I+T)=\ker(e^{-2i\theta}+\Lambda_{\theta}) \]
\end{lemma}
\begin{proof}
If $f\oplus g\in \ker L$, then $Af+Bg=0$ which implies that $f+T_{AB}g=0$ and $g+T_{BA}f=0$ after left-multiplying by $A^{\dagger}$ and $B^{\dagger}$, respectively. 
Then, $f\oplus g\in \ker(I+T)$. 
On the other hand, if $f\oplus g\in \ker(I+T)$, then $f+T_{AB}g=0$ and $g+T_{BA}f=0$, which implies that $A^{\dagger}(Af+Bg)=0$ and $B^{\dagger}(Af+Bg)=0$.
Then, $Af+Bg$ must be $0$, or equivalently, $f\oplus g\in \ker L$. 
We conclude that $\ker L=\ker(I+T)$. 

The second equality $\ker(I+T)=\ker(e^{-2i\theta}+\Lambda_{\theta})$ is obtained by applying the Gaussian elimination method to $(e^{-2i\theta}+\Lambda_{\theta})$. 
\end{proof}
Let us make a useful characterization of the eigenspace described by $\ker\left(\lambda-U|_{\mathcal{A}+\mathcal{B}}\right)$.
From Lemmas~\ref{lem:masterEq} and \ref{lem:kerL}, it follows that
	\begin{equation}\label{eq:Tachimachi} 
        \ker\left((\lambda-U_\theta)L\right)=\ker\left((I+T)(\lambda-\Lambda_\theta)\right)=\ker\left((e^{-2i\theta}+\Lambda_\theta)(\lambda-\Lambda_\theta)\right). 
        \end{equation}
%%%. 
Consider the nonreversible case. Corollary~\ref{lem:DBC_corol} states that in this case $\ker(I+T)=\{0\}$, that is, $I+T$ is invertible and 
Eq.~(\ref{eq:Tachimachi}) can be further reduced to 
	\[ \ker\left((e^{2i\phi}-U_\theta)L\right) = \ker\left(e^{2i\phi}-\Lambda_\theta\right)= D\ker \left(\frac{\cos\phi}{\sin\theta}-T\right). \]
Then, when $T$ is nonreversible, all eigenvalues of $U_\theta$ associated with the invariant subspace $\mathcal{A}+\mathcal{B}$ are 
obtained from the eigenvalues of $T$. 
%Indeed by Proposition \ref{thm: cycle-path}, for nonreversible case, we have 
%	 \[ \dim(\mathcal{A}+\mathcal{B})+\dim(\mathcal{A}+\mathcal{B})^\perp 
%         	= ({m}+{n})+(\nu-({m}+{n}))=\nu \]

Consider the reversible case. 
By (\ref{eq:Tachimachi}) and Lemma~\ref{lem:kerL}, if $e^{2i\phi}\neq -e^{-2i\theta}$, then 
	\begin{equation}\label{eq:Kasugamachi} 
        \ker\left((e^{2i\phi}-U_\theta)L\right)\setminus \ker L=\ker\left(e^{2i\phi}-\Lambda_\theta\right)=D\ker\left(\frac{\cos\phi}{\sin\theta}-T\right). 
        \end{equation}
Note that we obtain the eigenvalues of $U_\theta$ associated only with eigenvectors that do not belong to the kernel of $L$. Now we analyze the boundaries of the spectrum of $\Lambda_\theta$, which are $2\phi_+:=\pi-2\theta$ and $2\phi_-:=\pi+2\theta$. These boundaries exist only if $\theta\neq \pi/2$ as can be seen in Fig.~\ref{fig:mapping}. Still in the reversible case, we split the analysis into two cases.
\

\noindent{\bf Case $\theta\neq \pi/2$.} Counting the dimension of $\mathcal{A}+\mathcal{B}$ inherited from the eigenspace of $\Lambda_\theta$ except the eigenspace 
with the eigenvalue $e^{2i\phi_+}$, we have 
	\begin{align*}
        \sum_{\phi \neq \phi_+} \dim \ker\left(e^{2i\phi}-U_\theta|_{\mathcal{A}+\mathcal{B}}\right)
        	&= \sum_{\phi \neq \phi_+} \dim \left(D\ker\left(\frac{\cos\phi}{\sin\theta}-T\right)\right) \\
                &= |V_{\textrm{un}}|-\dim \ker(I-T) \\
                &= |V_{\textrm{un}}|-1.
        \end{align*}
On the other hand, since $T$ is reversible, then $\dim (\mathcal{A}+\mathcal{B})^\perp= b_1(G_{\textrm{un}})=|E_{\textrm{un}}|-|V_{\textrm{un}}|+1$ by (\ref{eq:Omachi}). 
If $e^{2i\phi_+}\in \sigma(U_\theta|_{\mathcal{A}+\mathcal{B}})$, then
$\sum_{\lambda}\dim(\ker(\lambda-U_\theta))>|E_{\textrm{un}}|=|V(G)|$, which is a contradiction. Then, $e^{2i\phi_+}\not\in \sigma(U_\theta|_{\mathcal{A}+\mathcal{B}})$.

\

\noindent{\bf Case $\theta=\pi/2$.} Eq.~(\ref{eq:Kasugamachi}) and the same results of the case $\theta\neq \pi/2$ hold in the present case, unless $\phi=0$, which implies that $\theta=\pi/2$.
When $\phi=0$, we have 
	\begin{equation} 
        \ker((I-U_{\pi/2})L)=\ker((I-\Lambda_{\pi/2})^2)=\ker((I-T)(I+T))=\ker(I-T)\oplus \ker L. 
        \end{equation}
The third expression is obtained by a Gaussian elimination and the final expression comes from Lemma~\ref{lem:kerL}. 
Using Lemma~\ref{lem:dimAB}, we obtain $\ker(I-U_{\pi/2}|_{\mathcal{A}+\mathcal{B}})=L\ker(I-T)=\mathbb{C}\,A\pi_1$.

Now we summarize the statements related to $\ker(\lambda-U|_{\mathcal{A}+\mathcal{B}})$:
\noindent
\begin{enumerate}
\item Non-reversible case: 
	\[ \ker(e^{2i\phi}-U|_{\mathcal{A}+\mathcal{B}})=LD\ker\left( \frac{\cos\phi}{\sin\theta}-T \right) \]
\item Reversible case:
\begin{enumerate}
\item If $\theta\neq \pi/2$, then
	\[ \ker(e^{2i\phi}-U|_{\mathcal{A}+\mathcal{B}})=\begin{cases}LD\ker\left( \frac{\cos\phi}{\sin\theta}-T \right) & \text{if $\phi\neq \phi_+$,} \\ 0 & \text{if $\phi=\phi_+$.} \end{cases} \]
\item If $\theta=\pi/2$, then
	\[ \ker(e^{2i\phi}-U|_{\mathcal{A}+\mathcal{B}})=\begin{cases}LD\ker\left( \cos\phi-T \right) & \text{if $\phi\neq 0$,} \\ \mathbb{C}A\pi_1 & \text{if $\phi=0$.} \end{cases} \]
\end{enumerate}
\end{enumerate}

%%%%
%%%%

%%%%%%%%%%%%%%
\subsection{Cycle-path space $(\mathcal{A}+\mathcal{B})^\perp$}\label{subsec:Cycle-path}
%%%%%%%%%%%%%%

In this subsection we address the subspace $\mathcal{A}^\perp \cap \mathcal{B}^{\perp}$, the dimension of which depends on the reversibility of $T$.  Using $$\dim(\mathcal{A}+\mathcal{B})=\dim(\mathcal{A})+\dim(\mathcal{B})-\dim(\mathcal{A}\cap\mathcal{B})$$  and Lemma~\ref{lem:dimAB}, we have 
	\begin{equation}
        \dim(\mathcal{A}^\perp \cap \mathcal{B}^\perp)
        =
        \begin{cases}
        \nu-{m}-{n}+1 & \text{if $T$ is reversible, }\\
        \nu-{m}-{n} & \text{otherwise. }
        \end{cases}
        \end{equation}
The dimension is expressed by the first Betti number $b_1(G_{\textrm{un}})$ of the underlying bipartite multigraph $G_{\textrm{un}}$, that is, 
	\begin{equation}\label{eq:Omachi}
        \dim(\mathcal{A}^\perp \cap \mathcal{B}^\perp)
        =
        \begin{cases}
        b_1(G_{\textrm{un}}) & \text{if $T$ is reversible, }\\
        b_1(G_{\textrm{un}}) - 1 & \text{otherwise, }
        \end{cases}
        \end{equation}
because the number of edges of $G_{\textrm{un}}$ is $\nu$ ($G$ is the line graph of $G_{\textrm{un}}$), the number of vertices of $G_{\textrm{un}}$ is $|\mathcal{T}_1|+|\mathcal{T}_2|=m+n$, and by definition $b_1(G_{\textrm{un}})=|E_{\textrm{un}}|-|V_{\textrm{un}}|+1$. In fact, the first Betti number is equal to the number of fundamental cycles. Here a fundamental cycle is the cycle generated by adding one edge of the original graph to the spanning tree. Since there is a one-to-one correspondence between the set of fundamental cycles and the set of edges not in the spanning tree, $b_1(G_{\textrm{un}})$ is equal to the number of edges of $G_{\textrm{un}}$ not in the spanning tree. The number of edges in the spanning tree of $G_{\textrm{un}}$ is $|V(G_{\textrm{un}})|-1=m+n-1$. Then, $b_1(G_{\textrm{un}})=\nu-m-n+1$.  

Let $\Gamma_{\textrm{un}}$ be a set of fundamental cycles of $G_{\textrm{un}}$. In the reversible case, there is a one-to-one correspondence between $\Gamma_{\textrm{un}}$ and a basis of the vector space $\mathcal{A}^\perp \cap \mathcal{B}^\perp$, which is isomorphic to the cycle space~\cite{GY05}. In Proposition~\ref{thm:cycle}, we show how to obtain an eigenvector of $U_\theta$ with eigenvalue $(-e^{-2i\theta})$ associated with a fundamental cycle.  In the nonreversible case, we have to fix one fundamental cycle and choose a second fundamental cycle and then we link these cycles with a path when they have no overlap forming a cycle-path subgraph. Since one cycle in $\Gamma_{\textrm{un}}$ remains fixed, this explains why there is a $(-1)$ in Eq.~(\ref{eq:Omachi}) in the nonreversible case. The vector space $\mathcal{A}^\perp \cap \mathcal{B}^\perp$ is not isomorphic to the cycle space.  In Proposition~\ref{thm:cycle-path}, we show how to obtain an eigenvector of $U_\theta$ with eigenvalue $(-e^{-2i\theta})$ associated with the cycle-path subgraph.

\begin{lemma}\label{lem:dimABperp}
If $\psi\in\mathcal{A}^\perp \cap \mathcal{B}^{\perp}$, then $\psi$ is an eigenfunction of $U_\theta$ with eigenvalue $(-e^{-2i\theta})$.
\end{lemma}
\begin{proof}
If $\psi\in\mathcal{A}^\perp \cap \mathcal{B}^{\perp}$, then $\psi$ is a $(e^{-i\theta})$-eigenvector of $e^{i\theta H_A}$ and  $e^{i\theta H_B}$ because $H_A$ and $H_B$ are unitary and self-adjoint operators. Then, $\psi$ is an eigenfunction of $U_\theta(=-e^{i\theta H_B}e^{i\theta H_A})$ with eigenvalue $(-e^{-2i\theta})$.
\end{proof}

Since $G$ is the line graph of $G_{\textrm{un}}$, there is a bijection map $\eta:E_{\textrm{un}}\to V(G)$, which we use in the following propositions.

\begin{proposition}\label{thm:cycle}
Suppose that $T$ is reversible.
Then, for each $c\in \Gamma_{\textrm{un}}$, 
there is an eigenfunction $\psi_c$ in $\mathcal{A}^\perp \cap \mathcal{B}^\perp$ 
whose support is $\mathrm{supp}(\psi_c) = \{ \eta(e) \;|\; e\in E(c) \}.$
Moreover, 
\begin{align}
\mathcal{A}^\perp \cap \mathcal{B}^\perp
	= \ker\left(e^{-2i\theta}+U_{\theta}|_{\mathcal{A}^\perp \cap \mathcal{B}^\perp}\right)=
        \mathrm{span}\{ \psi_c \;|\; c\in \Gamma_{\textrm{un}}\}.
\end{align}
\end{proposition}
\begin{proof}
Define $\mathcal{U}_\eta: \ell^2(E_{\textrm{un}})\to \ell^2(V(G))$ so that $\mathcal{U}_\eta(\psi)(u)=\psi(\eta^{-1}(u))$.
Note that $\mathcal{U}_\eta (\psi)\in \mathcal{A}^\perp$ if and only if 
\begin{equation}\label{eq:a_balance}
\sum_{e:\mathcal{T}_1(\eta(e))=\alpha}\overline{a}(\eta(e))\,\psi(e)=0
\end{equation}
for $\alpha\in \mathcal{T}_1$, where the sum runs over the edges incident to $\alpha\in \mathcal{T}_1$ in $G_{\textrm{un}}$ and $\overline{a}(\eta(e))$ is the complex conjugate of ${a}(\eta(e))$.
$\mathcal{U}_\eta(\psi)\in \mathcal{B}^\perp$ if and only if 
\begin{equation}\label{eq:b_balance}
\sum_{e:\mathcal{T}_2(\eta(e))=\beta}\overline{b}(\eta(e))\,\psi(e)=0
\end{equation}
for $\beta\in \mathcal{T}_2$, where the sum runs over the edges incident to $\beta\in \mathcal{T}_2$ in $G_{\textrm{un}}$. We use Eqs.~(\ref{eq:a_balance}) and~(\ref{eq:b_balance}) to obtain the entries of an eigenfuntion with support on a fundamental cycle.
From now on, we consider space $\ell^2(E_{\textrm{un}})$, which is lifted to $\ell^2(V(G))$ by the unitary map $\mathcal{U}_\eta$. For the sake of simplicity, we denote $a(\eta(e))$ by $a(e)$ and $b(\eta(e))$ by $b(e)$ for $e\in E_{\textrm{un}}$. 

Let $c\in\Gamma_{\textrm{un}}$. Since $G_{\textrm{un}}$ is a bipartite multigraph, $c$ is an even cycle. 
In the following, we describe an important property of a reversible measure on $V_{\textrm{un}}$ using this cycle and then we construct an eigenfunction $\psi_c\in \ell^2(E_{\textrm{un}})$ so that $\mathcal{U}_\eta(\psi_c)\in \mathcal{A}^\perp \cap \mathcal{B}^\perp$ and whose support is $E(c)=\{e_1,e_2,\dots,e_{2k}\}$. The vertices are labeled by $u_1=e_{1}\cap e_{2}$, $u_2=e_2\cap e_3$, $\cdots$, $u_{2k}=e_{2k}\cap e_{1}$ and $u_{2j}\in \mathcal{T}_1$ and $u_{2j-1}\in \mathcal{T}_2$ for $j=1,\dots, k$. 

Since $T$ is reversible, there is a reversible measure $\pi$ on $V_{\textrm{un}}$. 
Redefining $\pi$ so that $\pi(u_{2k})=1$ and considering the vertices $u_{2k}$ and $u_1$ connected by the edge $e_1$, 
the quantum detailed balance equation of Def.~\ref{def:DBC} implies that $a(e_1)\pi(u_{2k})=b(e_1)\pi(u_1)$, which simplifies to $\pi(u_1)=a(e_1)/b(e_1)$. 
Now we consider the neighboring vertices $u_1$ and $u_2$ connected by the edge $e_2$. We have $b(e_2)\pi(u_1)=a(e_2)\pi(u_2)$, which simplifies to $\pi(u_2)=b(e_2)a(e_1)/a(e_2)b(e_1)$. 
We proceed systematically considering the edges of the cycle $c$ until the final edge $e_{2k}$, which must satisfy 
	\begin{equation}\label{eq:cycleRev} 
        \frac{b(e_{2k})a(e_{2k-1})\cdots b(e_2)a(e_1)}{a(e_{2k})b(e_{2k-1})\cdots a(e_{2})b(e_1)}=\pi(u_{2k})=1. 
        \end{equation}

Now we describe the procedure that generates the eigenfunction $\mathcal{U}_\eta(\psi_c)\in \mathcal{A}^\perp\cap \mathcal{B}^\perp$ so that $\mathrm{supp}(\psi_c) =  E(c) .$
Setting $\psi_c(e_1)=1$ and considering the edges $e_1$ and $e_2$ with the common vertex $u_1$, 
Eq.~(\ref{eq:b_balance}) implies that $\bar{b}(e_1){\psi_c}(e_1)+\bar{b}(e_2){\psi_c}(e_2)=0$ because $\mathcal{U}_\eta(\psi_c)\in \mathcal{B}^\perp$, which simplifies to ${\psi_c}(e_2)=-\bar{b}(e_1)/\bar{b}(e_2)$. 
Next, since $\mathcal{U}_\eta(\psi_c)\in \mathcal{A}^\perp$, considering the edges $e_2$ and $e_3$ with the common vertex $u_2$, Eq.~(\ref{eq:a_balance}) implies that $\bar{a}(e_2){\psi_c}(e_2)+\bar{a}(e_3){\psi_c}(e_3)=0$, 
which simplifies to ${\psi_c}(e_3)=\bar{a}(e_2)\bar{b}(e_1)/\bar{a}(e_3)\bar{b}(e_2)$. 
We proceed systematically until the final vertex $u_{2k}$. 
Then, we close the cycle with no conflict because if we take one step further, we use (\ref{eq:cycleRev} and we obtain the consistency equation
	\[ \frac{\bar{a}(e_{2k})\bar{b}(e_{2k-1})\cdots \bar{a}(e_{2})\bar{b}(e_1)}{\bar{a}(e_1)\bar{b}(e_{2k})\cdots \bar{a}(e_{3})\bar{b}(e_2)}={\psi_c}(e_1)=1. \]
Summing up, the even entries of the eigenfunction are
\begin{equation}\label{eq:psie2j}
\psi_c(e_{2j})=-\frac{\bar{b}(e_{2j-1})\cdots \bar{a}(e_{2})\bar{b}(e_1)}{\bar{b}(e_{2j})\cdots \bar{a}(e_3)\bar{b}(e_2)},
\end{equation}
and the odd entries of the eigenfunction are
\begin{equation}\label{eq:psie2j+1}
\psi_c(e_{2j+1})=\frac{\bar{a}(e_{2j})\bar{b}(e_{2j-1})\cdots \bar{a}(e_{2})\bar{b}(e_1)}{\bar{a}(e_{2j+1})\bar{b}(e_{2j})\cdots \bar{a}(e_3)\bar{b}(e_2)}.
\end{equation}
Note that by construction 
$\mathcal{U}_\eta(\psi_c)\in \mathcal{A}^\perp \cap \mathcal{B}^\perp$ and $\mathrm{supp}(\mathcal{U}_\eta(\psi_c))= \{ \eta(e) \;|\; e\in E(c) \}.$
Using Lemma~\ref{lem:dimABperp}, we conclude that the above procedure generates a linearly independent eigenfunction with eigenvalue $(-e^{-2i\theta})$ for each fundamental cycle. Since the number of fundamental cycles is equal to $\dim(\mathcal{A}^\perp\cap \mathcal{B}^\perp)$ in the reversible case, the set of the eigenfunctions $\mathcal{U}_\eta(\psi_c)$ is an eigenbasis of $\mathcal{A}^\perp\cap \mathcal{B}^\perp$. 
\end{proof}

In the nonreversible case, we need to construct a graph using two fundamental cycles.

\

\noindent\textbf{Construction 1.} 
Let $c_0$ and $c$ be fundamental cycles. Define graph $G_c^{c_0}$, subgraph of $G_{\textrm{un}}$, obeying the following rules: (1)~If $V(c_0)\cap V(c)\neq\emptyset$, then $G_c^{c_0}=(V(c_0)\cup V(c),E(c_0)\cup E(c))$, that is, $G_c^{c_0}$ is the union of $c_0$ and $c$, and (2)~if $V(c_0)\cap V(c)= \emptyset$, then $G_c^{c_0}$ is the union of $c_0$, $c$, and a path $p$ connecting $c_0$ and $c$ so that $E(p)\cap (E(c_0)\cup E(c))=\emptyset$. 

\begin{proposition}\label{thm:cycle-path}
Suppose that $T$ is nonreversible.
Let $c_0$ and $c$ be cycles in $\Gamma_{\textrm{un}}$ and let $G_c^{c_0}$ be a graph obtained from Construction~1.
Then, for each $c\in \Gamma_{\textrm{un}}\setminus\{c_0\}$, there is an eigenfunction $\psi_c^{c_0}$ in $\mathcal{A}^\perp \cap \mathcal{B}^\perp$ whose support is $\mathrm{supp}(\psi_c^{c_0})= \{ \eta(e) \;|\; e\in E(G_c^{c_0}) \}.$
Moreover, 
\begin{align}
\mathcal{A}^\perp \cap \mathcal{B}^\perp
	= \ker\left(e^{-2i\theta}+U_{\theta}|_{\mathcal{A}^\perp \cap \mathcal{B}^\perp}\right)= 
        \mathrm{span}\{ \psi_c^{c_0} \;|\; c\in \Gamma_{\textrm{un}}\setminus\{c_0\}\}.
\end{align}
\end{proposition}
\begin{proof}
Let $c_0$ and $c$ be fundamental cycles.
The vertices of $c_0$ and $c$ are labeled by $\{ x_1,\dots,x_{2k} \}$ and $\{ y_1,\dots,y_{2k'} \}$, respectively, and the
edges of $c_0$ are $e_1=\{ x_{2k},x_1 \},\;e_2=\{ x_1,x_2 \},\dots,e_{2k}=\{ x_{2k-1},x_{2k} \}$, and the edges of $c$ are
$f_1=\{ y_{2k'},y_1 \},\;f_2=\{ y_1,y_2 \},\dots,f_{2k'}=\{ y_{2k'-1},y_{2k'} \}$,  where $x_{2j}\in \mathcal{T}_1$ and $x_{2j-1}\in \mathcal{T}_2$. Let $\mathcal{U}_\eta$ be the operator defined in the proof of Proposition~\ref{thm:cycle}.

\

\noindent{\bf Case 1.}
Suppose that $V(c_0)\cap V(c)= \emptyset$. Then, there is a path $p$ connecting the two cycles from $x_{2k}\in V(c_0)$ to $y_{2k'}\in V(c)$. 
Denote the vertices of the path by $\{ x_{2k}=z_1,z_2,\dots,z_{\ell+1}=y_{2k'} \}$ and the edges by 
$g_1=\{z_1,z_2\},g_2=\{z_2,z_3\},\dots,g_{\ell}=\{z_{\ell},z_{\ell+1}\}$. 

Now we construct function $\psi_c^{c_0}$ whose support is $E(G_c^{c_0})$. Let us define balancing indices $\Delta(c_0)$ and $\Delta(c)$ for cycles $c_0$ and $c$, so that 
\begin{align*} 
\Delta(c_0) &= \frac{a(e_{2k})b(e_{2k-1})\cdots a(e_{2})b(e_1)}{b(e_{2k})a(e_{2k-1})\cdots b(e_2)a(e_{1})}-1, \\
\Delta(c) &= \frac{a(f_{2k'})b(f_{2k'-1})\cdots a(f_{2})b(f_1)}{b(f_{2k'})a(f_{2k'-1})\cdots b(f_2)a(f_{1})}-1. 
\end{align*}
Note that in the reversible case we would have $\Delta(c_0)=\Delta(c)=0$. Since we are addressing the nonreversible case, we do have $\Delta(c_0)=\Delta(c)\neq 0$. 
Using the same procedure of Proposition~\ref{thm:cycle}, we start from vertex $x_1$ with $\psi_c^{c_0}(e_1)=1$ and proceed until $x_{2k-1}$.  The entries of $\psi_c^{c_0}$ on $c_0$ are given by Eqs.~(\ref{eq:psie2j}) and~(\ref{eq:psie2j+1}).

The last vertex $x_{2k}$ has three incident edges in $G_c^{c_0}$ and Eq.~(\ref{eq:a_balance}) implies that $\bar{a}(e_1)\psi_c^{c_0}(e_1)+\bar{a}(e_{2k})\psi_c^{c_0}(e_{2k})+\bar{a}(g_1)\psi_c^{c_0}(g_1)=0$ because $\mathcal{U}_\eta(\psi_c^{c_0})\in \mathcal{A}^\perp$. Using the expression of $\psi_c^{c_0}(e_{2k})$ and the balancing index $\Delta(c_0)$, we obtain
\begin{equation}\label{eq:psi_g1}
\bar{\psi}_c^{c_0}(g_1)=\frac{a(e_1)}{a(g_1)}\Delta(c_0).
\end{equation}
Continuing to use the procedure from $z_2$ going along the path $p$ until $z_{\ell}$, 
we obtain 
\begin{equation}\label{eq:g2l}
\bar{\psi}_c^{c_0}(g_{\ell})=
\begin{cases} -\frac{\kappa(p)a(e_1)\Delta(c_0)}{a(g_{\ell})} &\mbox{if } \ell \mbox{ is even,}  
\vspace{7pt} \\
+\,\frac{\kappa(p)a(e_1)\Delta(c_0)}{b(g_{\ell})}\frac{}{} & \mbox{if } \ell \mbox{ is odd,} \end{cases}
\end{equation}
where
\[ 
\kappa(p)=
\begin{cases} \frac{a(g_{\ell})b(g_{\ell-1})\cdots a(g_2)b(g_1)}{b(g_{\ell})a(g_{\ell-1})\cdots b(g_2)a(g_1)} &\mbox{if } \ell \mbox{ is even,}  
\vspace{7pt} \\
\frac{b(g_{\ell})a(g_{\ell-1})\cdots a(g_2)b(g_1)}{a(g_{\ell})b(g_{\ell-1})\cdots b(g_2)a(g_1)} & \mbox{if } \ell \mbox{ is odd.} \end{cases} \]
%%%\[ \kappa(p)=\frac{a(g_{\ell})b(g_{\ell-1})\cdots a(g_2)b(g_1)}{b(g_{\ell})a(g_{\ell-1})\cdots b(g_2)a(g_1)}. \]

Now we restart the procedure in order to find the entries of $\psi_c^{c_0}$ on the cycle $c$. We start from $y_1$ with an arbitrary $\psi_c^{c_0}(f_1)$ and proceed until $y_{2k'-1}$. The goal is to obtain $\psi_c^{c_0}(f_1)$ that consistently closes the process. Suppose that the length $\ell$ of path $p$ is even. Then, $y_{2j}\in \mathcal{T}_1$ and $y_{2j-1}\in \mathcal{T}_2$. In this case, the tessellations of cycles $c_0$ and $c$ are symmetric and we can use the result of Eq.~(\ref{eq:psi_g1}) by replacing $g_1$ by $g_{\ell}$ and  edges $e$ by edges $f$, that is, 
\begin{equation}\label{eq:psig_ell}
 \bar{\psi}_c^{c_0}(g_{\ell})=\frac{a(f_1)}{a(g_{\ell})}\Delta(c)\bar{\psi}_c^{c_0}(f_1). 
 \end{equation} 
In order to obtain a consistent result, we use Eq.~(\ref{eq:g2l}) (even $\ell$) to determine $\bar{\psi}_c^{c_0}(f_1)$, which is the only one missing. The result is 
	\begin{equation}\label{eq:condition}
        \bar{\psi}_c^{c_0}(f_1)=-\kappa(p)\frac{a(e_1)\Delta(c_0)}{a(f_1)\Delta(c)}.
        \end{equation}

Suppose that the length of path $p$ is odd, where the length is the number of edges. Then, $y_{2j-1}\in \mathcal{T}_1$ and $y_{2j}\in \mathcal{T}_2$. In this case, the tessellations of cycles $c_0$ and $c$ are antisymmetric and we have to recalculate Eq.~(\ref{eq:psig_ell}). The new result is  
	\[ \bar{\psi}_c^{c_0}(g_{\ell})=\frac{b(f_1)}{b(g_{\ell})}\tilde{\Delta}(c)\bar{\psi}_c^{c_0}(f_1), \]
where $\tilde{\Delta}(c)$ is obtained from $\Delta(c)$ by interchanging $a$ and $b$.
Again, in order to obtain a consistent result, we use Eq.~(\ref{eq:g2l}) (odd $\ell$) to determine the $\bar{\psi}_c^{c_0}(f_1)$. The result is 
	\begin{equation}\label{eq:condition2}
        \bar{\psi}_c^{c_0}(f_1)=\kappa(p)\frac{a(e_1)\Delta(c_0)}{b(f_1)\tilde{\Delta}(c)}.
        \end{equation}
By construction $\mathcal{U}_\eta(\psi_c^{c_0})\in \mathcal{A}^\perp \cap \mathcal{B}^\perp$ and $\mathrm{supp}(\psi_c^{c_0})=  E(G_c^{c_0}) .$

\

\noindent{\bf Case 2.}
Suppose that $V(c_0)\cap V(c)\neq \emptyset$. \red{In this case, the intersection of the two fundamental cycles of $G_c^{c_0}$ is either a vertex or a path~\cite{Hoe92}.}
%\red{if not, that is, the intersection consists of disjoint union of paths, then it is contradiction to the definition of the fundamental cycle.}
The path length is odd if $|V(c_0)\cap V(c)|$ is even and the path length is even if $|V(c_0)\cap V(c)|$ is odd. In both situations, we start with the cycle $c_0$ using the same procedure of Case 1, but this time there are two degree-3 vertices (or one degree-4 if $|V(c_0)\cap V(c)|=1$). As before, we assume that $\psi^{c_0}(e_1)=1$ and proceed along the edges of $c_0$. Then, we restart the procedure with cycle $c$ with arbitrary $\psi_c(f_1)$ and proceed along $c$.  The entries of $\psi^{c_0}$ on $c_0$ are given by Eqs.~(\ref{eq:psie2j}) and~(\ref{eq:psie2j+1}) and  the entries of $\psi_c$ on $c$ are given by Eqs.~(\ref{eq:psie2j}) and~(\ref{eq:psie2j+1}) after multiplying the right-hand side by $\psi_c(f_1)$ and after interchanging $e_{2j}$ and $e_{2j+1}$ by $f_{2j}$ and $f_{2j+1}$, respectively. There is an identification of edges of $c_0$ and $c$ along the path. The eigenvector with support on the union of the cycles is $\psi_c^{c_0}=\psi^{c_0}+\psi_c$. In order to close the set of all equations based on Eqs.~(\ref{eq:a_balance}) and~(\ref{eq:b_balance}), we have to set
\begin{equation}\label{Eq:psi_c(f1)}
\bar{\psi}_c(f_1)=-\frac{a(e_1)\Delta(c_0)}{a(f_1)\Delta(c)},
\end{equation}
which is the same as the one given by Eq.~(\ref{eq:condition}) if we take $\kappa(p)=1$.

Note that by construction 
$\mathcal{U}_\eta(\psi_c^{c_0})\in \mathcal{A}^\perp \cap \mathcal{B}^\perp$ and $\mathrm{supp}(\mathcal{U}_\eta(\psi_c^{c_0}))= \{ \eta(e) \;|\; e\in E(G_c^{c_0}) \}.$
Using Lemma~\ref{lem:dimABperp}, we conclude that the above procedure generates a linearly independent eigenfunction with eigenvalue $(-e^{-2i\theta})$ for each fundamental cycle $c\in \Gamma_{\textrm{un}}\setminus\{c_0\}$. Since the number of fundamental cycles in $\Gamma_{\textrm{un}}\setminus\{c_0\}$ is equal to $\dim(\mathcal{A}^\perp\cap \mathcal{B}^\perp)$ in the nonreversible case, the set of the eigenfunctions $\mathcal{U}_\eta(\psi_c^{c_0})$ is an eigenbasis of $\mathcal{A}^\perp\cap \mathcal{B}^\perp$. 
\end{proof}

\section{Main theorem}\label{sec:maintheo}

Suppose we have defined a 2-tessellable quantum walk as described in Section~\ref{Sec:SQW}. We also assume that the quantities and the notations of Secs.~\ref{eigvals} and~\ref{eigvecs} are known. For instance, $\phi$ and $\phi_\pm$ are defined as $\cos\phi(\mu)=\mu\sin\theta$ and $\phi_{\pm}=\pi/2\mp \theta$. 
We summarize our results in the following theorem. 
\begin{theorem}\label{thm:main}
$U_\theta$ can be decomposed into $U_\theta|_{\mathcal{A}+\mathcal{B}}\oplus U_\theta|_{(\mathcal{A}+\mathcal{B})^{\perp}}$ 
under the decomposition of $\mathcal{H}=(\mathcal{A}+\mathcal{B})\oplus (\mathcal{A}+\mathcal{B})^\perp$. 
Each invariant space is further decomposed as follows.
\begin{enumerate}
\item Space $\mathcal{A}+\mathcal{B}$. \\
\begin{align*} 
\sigma(U_\theta|_{\mathcal{A}+\mathcal{B}}) 
	&= \begin{cases}
        \{e^{2i\phi(\mu)} \;|\; \mu\in \sigma(T)\} & \text{if $T$ is nonreversible,}\\
        \{e^{2i\phi(\mu)} \;|\; \mu\in \sigma(T)\setminus\{1\}\} & \text{if $T$ is reversible and $\theta\neq \pi/2$,}\\
        \{e^{2i\phi(\mu)} \;|\; \mu\in \sigma(T)\setminus\{-1\}\} & \text{if $T$ is reversible and $\theta= \pi/2$.}
          \end{cases} 
\end{align*}
	\begin{multline*} 
	\ker(U_\theta|_{\mathcal{A}+\mathcal{B}}-e^{2i\phi})=
	\\ \hspace{0.5cm}\begin{cases} 
        \left\{Af+ie^{i(\theta+\phi)}Bg \;|\; f\oplus g\in \ker\left(\frac{\cos\phi}{\sin\theta} -T\right)\right\} & \text{if $\phi\neq \phi_-$,}\\
        \mathcal{A}\cap\mathcal{B}=\mathbb{C}A\pi_1 & \text{if $\phi=\phi_-$ and $T$ is reversible, }\\
        0 & \text{if $\phi=\phi_-$ and $T$ is nonreversible,}
        \end{cases}
\end{multline*}
where $\pi_1\oplus \pi_2\in \ker(1-T)$, which is a reversible measure of $T$.
\item Space $(\mathcal{A}+\mathcal{B})^\perp$.\\
\begin{equation*}
        \sigma(U_\theta|_{(\mathcal{A}+\mathcal{B})^\perp})
        	=\begin{cases}
                \emptyset & \text{if $G_{\textrm{un}}$ is a tree or ``$b(G_{\textrm{un}})=1$ and $T$ is nonreversible'',}\\
                -e^{-2i\theta} & \text{otherwise.}
                \end{cases}
        \end{equation*}
%where $b_1(G_{\textrm{un}})$ is the first Betti number of $G_{\textrm{un}}$, that is, $b_1(G_{\textrm{un}})=|E_{\textrm{un}}|-|V_{\textrm{un}}|+1$. The eigenspace is generated by all the fundamental cycles of $G_{\textrm{un}}$ as follows: setting $\Gamma(G_{\textrm{un}})$ as a set of all the fundamental cycles of $G_{\textrm{un}}$, then we have 
\begin{equation*}
        \ker(U_\theta|_{(\mathcal{A}+\mathcal{B})^\perp}+e^{-2i\theta}) 
        	= \begin{cases}
                \mathrm{span}\{ \psi_{c} \;| \; c\in \Gamma(G_{\textrm{un}}) \} & \text{if $T$ is reversible,} \\
                \mathrm{span}\{ \psi_{c}^{c_0} \;| \;c\in \Gamma(G_{\textrm{un}})\setminus\{c_0\} \} & \text{otherwise.}
                \end{cases}
        \end{equation*}
If $T$ is reversible, then the dimension of $(\mathcal{A}+\mathcal{B})^\perp$ is $\nu-({m}+{n})+1$, 
otherwise, $\nu-({m}+{n})$. 
\end{enumerate}
\end{theorem}
%\begin{remark}
Note that the eigenvalues of $U_\theta$ were obtained via two different methods. The first method is described in Section~\ref{eigvals} and the eigenvalues are listed in Corollary~\ref{cor:eigenvalues}, items~(\ref{eq:062712}) and~(\ref{eq:062711}). The second method is described in Section~\ref{eigvecs} and were obtained using Gaussian elimination method. Note that those methods are consistent with each other.
%\end{remark}

%%%%%%%%%%%
%%%%%%%%%%%
\section{Example: kagome lattice}\label{kagome}
The kagome (or trihexagonal) lattice~\cite{Syo51,Sun13,KoShSu} is known to be the line graph of the hexagonal lattice and it is straightforward to check that the hexagonal lattice is the clique graph of the kagome lattice. The kagome lattice is 2-tessellable because the hexagonal lattice is 2-colorable, as can be checked in Fig.~\ref{Fig2}(a).
Fig.~\ref{Fig2}(a) also shows how we have embedded the kagome lattice in $\mathbb{R}^2$. For each horizontal line, there are upper and lower triangles at each crossing with a vertical line.  Take the lower red triangle in the center of Fig.~\ref{Fig2}(a), whose vertices are labeled by $\{1,2,3\}$, as a representative cell at $(x,y)\in \mathbb{Z}^2$. Then, each vertex of the kagome lattice is represented by $\mathbb{Z}^2\times\{1,2,3\}$.  Fig.~\ref{Fig2}(b) shows the underlying graph, which is the hexagonal lattice.

%Figure
\begin{figure}[!ht]
\begin{center}
	\includegraphics[width=100mm]{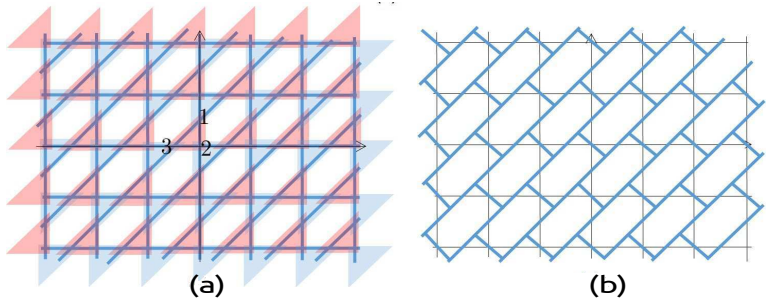}
\end{center}
\caption{ Graph (a) depicts the embedding of the kagome lattice and a tessellation cover, where $\mathcal{T}_1$ is comprises the red tiles
and $\mathcal{T}_2$ the blue tiles. Note the labeling of the red clique in the center. Graph~(b) describes the underlying bipartite graph, which is the hexagonal lattice.
}
\label{Fig2}
\end{figure}

Let $\mathcal{T}_1$ and $\mathcal{T}_2$ be the tessellations comprising the upper triangles and the lower triangles, respectively.
Define the $3$-dimensional self-adjoint unitary operators associated with each clique $E_1=2|\alpha\rangle\langle \alpha|-1$ and $E_2=2|\beta\rangle\langle \beta|-1$, where $|\alpha\rangle$ and $|\beta\rangle$ are unit vectors in $\mathbb{C}^3$. Set $H_1:=\oplus_{\mathcal{T}_1}E_1$ and $H_2:=\oplus_{\mathcal{T}_2}E_2$. Then, the evolution operator is $U_\theta=-e^{i\theta H_2}e^{i\theta H_1}$.

Due to the translational symmetries of the kagome lattice, we can use the Fourier transform $\mathcal{F}: \ell^2(\mathbb{Z}^2\times \{1,2,3\})\to L^2([0,2\pi)^2 \times \{1,2,3\})$, which is defined by 
$$\hat{\psi}(k,l,j):=(\mathcal{F}\psi)(k,l,j)=\sum_{x,y\in \mathbb{Z}}\psi(x,y,j)e^{i(kx+ly)},$$
where $j\in\{1,2,3\}$, and its inverse by 
\[\psi(x,y,j):=(\mathcal{F}^{-1}\hat{\psi})(x,y,j)=\frac{1}{(2\pi)^2}\int_{0}^{2\pi}\int_{0}^{2\pi}\hat{\psi}(k,l,j) e^{-i(kx+ly)}\text{d}k\text{d}l.\] 
In the Fourier space, the dynamic is described by a reduced $3\times 3$ evolution operator 
	\begin{equation}
        \hat{U}_\theta(k,l):=-{W}_{k,l}e^{i\theta E_2}{W}_{k,l}^{\dagger} e^{i\theta E_1}=-e^{i\theta E_2'(k,l)}e^{i\theta E_1},
        \end{equation}
where $W_{k,l}=\mathrm{diag}(1, e^{il},e^{i(k+l)})$ and $E_2'(k,l)=2|\beta'(k,l)\rangle\langle\beta'(k,l)|-1$ with $|\beta'(k,l)\rangle=W_{k,l}|\beta\rangle$, which can be obtained directly from the quotient graph of the kagome lattice (see left-hand graph of Fig.~\ref{Fig3}). A tessellation cover of the quotient graph comprises two polygons, which are associated with vectors $|\alpha\rangle$ and $|\beta'(k,l)\rangle$. The intersection graph of the quotient graph is the right-hand graph of Fig.~\ref{Fig3}.
The discriminant operator of $\hat{U}_\theta(k,l)$ for $|\alpha\rangle=|\beta\rangle=1/\sqrt{3}[1\;1\;1]^\dagger$ is given by 
\begin{align*}
\hat{T}(k,l) &= \begin{bmatrix} 0 & \langle \alpha| \beta'(k,l)\rangle \\ \langle \beta'(k,l)|\alpha\rangle & 0 \end{bmatrix} \\
	&= \frac{1}{3} \begin{bmatrix} 0 & 1+e^{il}+e^{i(k+l)} \\ 1+e^{-ik}+e^{-i(k+l)} & 0 \end{bmatrix}.
\end{align*}
 $\hat{T}(k,l)$ is defined on the intersection graph and its spectrum is $\sigma(\hat{T}(k,l))=\{\pm|1+e^{il}+e^{i(k+l)}|/3\}$. Using Theorem~\ref{thm:main}, we obtain 
	\begin{align*} 
        &\sigma(\hat{U}_\theta(k,l)|_{\mathcal{A}+\mathcal{B}})=\Big\{e^{i\eta} \;|\; \cos \eta
        	=-1+\frac{2}{3\sin^2\theta}+\frac{4}{9\sin^2\theta}(\cos l+\cos k+\cos(k+l))\Big\}, \\
	&\sigma(\hat{U}_\theta(k,l)|_{(\mathcal{A}+\mathcal{B})^\perp}
        	=\{-e^{-2i\theta}\}.
        \end{align*}

%Figure
\begin{figure}[!ht]
\begin{center}
\begingroup%
  \makeatletter%
  \providecommand\color[2][]{%
    \errmessage{(Inkscape) Color is used for the text in Inkscape, but the package 'color.sty' is not loaded}%
    \renewcommand\color[2][]{}%
  }%
  \providecommand\transparent[1]{%
    \errmessage{(Inkscape) Transparency is used (non-zero) for the text in Inkscape, but the package 'transparent.sty' is not loaded}%
    \renewcommand\transparent[1]{}%
  }%
  \providecommand\rotatebox[2]{#2}%
  \newcommand*\fsize{\dimexpr\f@size pt\relax}%
  \newcommand*\lineheight[1]{\fontsize{\fsize}{#1\fsize}\selectfont}%
  \ifx\svgwidth\undefined%
    \setlength{\unitlength}{341.08496781bp}%
    \ifx\svgscale\undefined%
      \relax%
    \else%
      \setlength{\unitlength}{\unitlength * \real{\svgscale}}%
    \fi%
  \else%
    \setlength{\unitlength}{\svgwidth}%
  \fi%
  \global\let\svgwidth\undefined%
  \global\let\svgscale\undefined%
  \makeatother%
  \begin{picture}(1,0.41250738)%
    \lineheight{1}%
    \setlength\tabcolsep{0pt}%
    \put(0,0){\includegraphics[width=\unitlength]{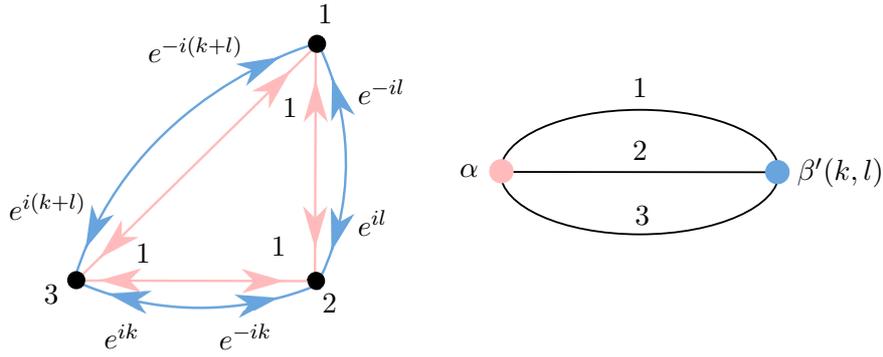}}%
    \put(0.70111719,0.29284137){\color[rgb]{0,0,0}\makebox(0,0)[lt]{\lineheight{1.25}\smash{\begin{tabular}[t]{l}$1$\end{tabular}}}}%
    \put(0.70111719,0.2237254){\color[rgb]{0,0,0}\makebox(0,0)[lt]{\lineheight{1.25}\smash{\begin{tabular}[t]{l}$2$\end{tabular}}}}%
    \put(0.70362205,0.15266755){\color[rgb]{0,0,0}\makebox(0,0)[lt]{\lineheight{1.25}\smash{\begin{tabular}[t]{l}$3$\end{tabular}}}}%
    \put(0.51243556,0.20541224){\color[rgb]{0,0,0}\makebox(0,0)[lt]{\lineheight{1.25}\smash{\begin{tabular}[t]{l}$\alpha$\end{tabular}}}}%
    \put(0.88112206,0.20101451){\color[rgb]{0,0,0}\makebox(0,0)[lt]{\lineheight{1.25}\smash{\begin{tabular}[t]{l}$\beta'(k,l)$\end{tabular}}}}%
    \put(0.35711967,0.37425724){\color[rgb]{0,0,0}\makebox(0,0)[lt]{\lineheight{1.25}\smash{\begin{tabular}[t]{l}$1$\end{tabular}}}}%
    \put(0.36202021,0.05538975){\color[rgb]{0,0,0}\makebox(0,0)[lt]{\lineheight{1.25}\smash{\begin{tabular}[t]{l}$2$\end{tabular}}}}%
    \put(0.05768683,0.06584591){\color[rgb]{0,0,0}\makebox(0,0)[lt]{\lineheight{1.25}\smash{\begin{tabular}[t]{l}$3$\end{tabular}}}}%
    \put(0.30634089,0.11810996){\color[rgb]{0,0,0}\makebox(0,0)[lt]{\lineheight{1.25}\smash{\begin{tabular}[t]{l}$1$\end{tabular}}}}%
    \put(0.15806503,0.10988871){\color[rgb]{0,0,0}\makebox(0,0)[lt]{\lineheight{1.25}\smash{\begin{tabular}[t]{l}$1$\end{tabular}}}}%
    \put(0.31885875,0.27088744){\color[rgb]{0,0,0}\makebox(0,0)[lt]{\lineheight{1.25}\smash{\begin{tabular}[t]{l}$1$\end{tabular}}}}%
    \put(0.40095022,0.28471121){\color[rgb]{0,0,0}\makebox(0,0)[lt]{\lineheight{1.25}\smash{\begin{tabular}[t]{l}$e^{-i l}$\end{tabular}}}}%
    \put(0.40089493,0.14060529){\color[rgb]{0,0,0}\makebox(0,0)[lt]{\lineheight{1.25}\smash{\begin{tabular}[t]{l}$e^{i l}$\end{tabular}}}}%
    \put(0.12341075,0.01584466){\color[rgb]{0,0,0}\makebox(0,0)[lt]{\lineheight{1.25}\smash{\begin{tabular}[t]{l}$e^{i k}$\end{tabular}}}}%
    \put(0.24912241,0.01497775){\color[rgb]{0,0,0}\makebox(0,0)[lt]{\lineheight{1.25}\smash{\begin{tabular}[t]{l}$e^{-i k}$\end{tabular}}}}%
    \put(0.01943037,0.1577393){\color[rgb]{0,0,0}\makebox(0,0)[lt]{\lineheight{1.25}\smash{\begin{tabular}[t]{l}$e^{i (k+l)}$\end{tabular}}}}%
    \put(0.17131814,0.33095402){\color[rgb]{0,0,0}\makebox(0,0)[lt]{\lineheight{1.25}\smash{\begin{tabular}[t]{l}$e^{-i (k+l)}$\end{tabular}}}}%
  \end{picture}%
\endgroup%
\end{center}
\caption{The left-hand graph is the quotient graph of the kagome lattice, which is  is a triangle with double edges. We have depicted the weights of the red and blue arcs.
The right-hand graph is the intersection graph of the quotient graph. The labels of the edges are $1$, $2$, and $3$. 
}
\label{Fig3}
\end{figure}

The eigenvectors of $\hat{U}_\theta(k,l)$ in $(\mathcal{A}+\mathcal{B})$ are obtained from the eigenvectors of $\hat{T}(k,l)$ and the eigenvectors in $(\mathcal{A}+\mathcal{B})^\perp$ are obtained using the cycle-path method. Let us focus on the latter.        
Let the set of edges of the quotient graph be $\{1,2,3\}$ as depicted in Fig.~\ref{Fig3}. The labeling can be converted to the notation of Section~\ref{subsec:Cycle-path} by using  $\{e_1,e_2,f_1\}$ and by identifying $f_2$ with $e_2$. Since $\hat{T}(k,l)$ does not have $(+1)$-eigenvectors when $(k,l)\neq (0,0)$, $\hat{T}(k,l)$ is nonreversible. The spanning tree of the intersection graph contains only one edge (take the one with label $2$). There are two fundamental cycles: $c_0=\{1,2\}$ and $c=\{2,3\}$. We use Case~2 of Proposition~\ref{thm:cycle-path} in order to compute an eigenvector $\hat{\psi}_c^{c_0}\in \ker(e^{-2i\theta}+\hat{U}_\theta)$. We have $a(1)=a(2)=e(3)=1/\sqrt{3}$ and $b(1)=1/\sqrt{3}$,  $b(2)=e^{il}/\sqrt{3}$, $b(3)=e^{i(k+l)}/\sqrt{3}$. We set $\hat{\psi}^{c_0}(1)=1$, and using Eq.~(\ref{eq:psie2j}) we obtain $\hat{\psi}^{c_0}(2)=-e^{il}$. Using Eq.~(\ref{Eq:psi_c(f1)}), we obtain $\hat{\psi}_c(3)=(e^{il}-1)/(1-e^{-ik})$ and then $\hat{\psi}_c(2)=(1-e^{il})/(e^{ik}-1)$. Adding vectors $\hat{\psi}^{c_0}$ and $\hat{\psi}_c$, we obtain $\hat{\psi}_c^{c_0}=(1, (e^{-ik}-e^{il})/(1-e^{-ik}), (e^{il}-1)/(1-e^{-ik}))$, which is an eigenvector of $\hat{U}_\theta(k,l)$ with eigenvalue $(-e^{-2i\theta})$.

By applying the inverse Fourier transform, the eigenvectors, which have finite support, are lifted up to the real space of the original graph $G$ and are expressed by 
	\[ \psi_{(x,y)}(x',y',j)=
        	\begin{cases}
                1 & \text{if $(x',y',j)=\{(x,y,1),(x,y+1,3),(x-1,y,2)\}$,} \\
                -1 & \text{if $(x',y',j)=\{(x,y,3),(x,y+1,2),(x-1,y,1)\}$,} \\
                0 & \text{otherwise,}
                \end{cases} \]
for any $(x,y)\in \mathbb{Z}^2$. The support of $\psi_{(x,y)}$ is a $6$-cycle of the kagome lattice. 
Note that if the initial state has overlap with any of these eigenvectors, there will be localization.

%%%%%%%
%%%%%%%%%%%
%%%%%%%%%%%
%\section{Conclusion}
%%%%%%%%%%%
%%%%%%%%%%%

%%%%%%%%%%%%%%%%%%%%%%%%%%%%%%%%%%%%%%%%%%%%%%%%%%%%%%%%%%%%%%%%%%%%%%%%%%%%%%%%%%%%%%%%%%%%%%%%%%%%%%%%%%%
%%%%%%%%%%%%%%%%%%%%%%%%%%%%%%%%%%%%%%%%%%%%%%%%%%%%%%%%%%%%%%%%%%%%%%%%%%%%%%%%%%%%%%%%%%%%%%%%%%%%%%%%%%%

%Figure
%\begin{figure}[htbp]
%\begin{center}
%	\includegraphics[width=150mm]{}
%\end{center}
%\caption{}
%\label{fig:one}
%\end{figure}

\section*{Acknowledgments}
YuH's work was supported in part by Japan Society for the Promotion of Science Grant-in-Aid for Scientific Research 
(C) 25400208, (C) 18K03401 and (A) 15H02055. RP is grateful to the kind hospitality of the Graduate School of Information Sciences and Research Alliance Center for Mathematical Sciences (RACMaS), Tohoku University, which sponsored his visit during the winter of 2018.
IS is partially supported by the Grant-in-Aid for Scientific Research (C) of Japan Society for the Promotion of Science (Grant No.~15K04985). 
ES acknowledges financial support from the Grant-in-Aid for Young Scientists (B) and of Scientific Research (B) Japan Society for the Promotion of Science (Grant No. 16K17637, No. 16K03939).
\par

%\appendix
%\def\thesection{Appendix \Alph{section}}
%\renewcommand{\theequation}{A.\arabic{equation}}
%\setcounter{equation}{0}

%\section{}
% Appendix here

%\bibliography{Bibliografia}

\begin{thebibliography}{10}


\bibitem{AAKV01}
D.~Aharonov, A.~Ambainis, J.~Kempe, and U.~Vazirani.
\newblock Quantum walks on graphs.
\newblock In {\em Proc. 33th STOC}, pages 50--59, New York, 2001. ACM.

\bibitem{Amb04}
A.~Ambainis.
\newblock Quantum walk algorithm for element distinctness.
\newblock In {\em Proc. 45th Annual IEEE Symposium on Foundations of Computer
  Science FOCS}, pages 22--31, Washington, 2004.

\bibitem{CMQKSPDSSBM15}
F.~Cardano, F.~Massa, H.~Qassim, E.~Karimi, S.~Slussarenko, D.~Paparo,
  C.~De~Lisio, F.~Sciarrino, E.~Santamato, R.W. Boyd, and L.~Marrucci.
\newblock Quantum walks and wavepacket dynamics on a lattice with twisted
  photons.
\newblock {\em Science Advances}, 1(2):e1500087, 2015.

\bibitem{CBS10}
C. M. Chandrashekar, S. Banerjee, and R. Srikanth.
\newblock Relationship between quantum walks and relativistic quantum mechanics.
\newblock {\em Phys. Rev. A}, 81:062340, 2010

\bibitem{Chi10}
A.~M. Childs.
\newblock On the relationship between continuous- and discrete-time quantum
  walk.
\newblock {\em Commun. Math. Phys.}, 294(2):581--603, 2010.

\bibitem{FG98}
E.~Farhi and S.~Gutmann.
\newblock Quantum computation and decision trees.
\newblock {\em Phys. Rev. A}, 58:915--928, 1998.

\bibitem{GASWWMA13}
M.~Genske, W.~Alt, A.~Steffen, A.~H. Werner, R.~F. Werner, D.~Meschede, and
  A.~Alberti.
\newblock Electric quantum walks with individual atoms.
\newblock {\em Phys. Rev. Lett.}, 110:190601, 2013.

\bibitem{GY05}
J.~L. Gross and J.~Yellen.
\newblock {\em Graph Theory and Its Applications}.
\newblock Chapman \& Hall/CRC, Boca Raton, FL, 2005.

\bibitem{Gudder}
S.~Gudder.
\newblock Quantum Markov chain. 
\newblock {\em J. Math. Phys.}, 49(7):072105, 2008.

\bibitem{HKSS14}
Yu.~Higuchi, N.~Konno, I.~Sato, and E.~Segawa.
\newblock Spectral and asymptotic properties of {G}rover walks on crystal lattices.
\newblock {\em J. Funct. Anal.}, 267(11):4197--4235, 2014.

\bibitem{HS04} Yu.~Higuchi and T.~Shirai.
\newblock Some spectral and geometric properties for infinite graphs.
\newblock {\em Contemp. Math.} 347:29--56, 2004.

\bibitem{Hoe92}
C.~Hoede.
\newblock A characterization of consistent marked graphs.
\newblock {\em J. Graph Theory}, 16(1):17--23, 1992.

\bibitem{KRBD10}
T.~Kitagawa, M.~S. Rudner, E.~Berg, and E.~Demler.
\newblock Exploring topological phases with quantum walks.
\newblock {\em Phys. Rev. A}, 82:033429, 2010.

\bibitem{KIS18}
N.~Konno, Y.~Ide, and I.~Sato.
\newblock The spectral analysis of the unitary matrix of a 2-tessellable
  staggered quantum walk on a graph.
\newblock {\em Linear Algebra Appl.}, 545:207--225, 2018.

\bibitem{KPSS18}
N.~Konno, R.~Portugal, I.~Sato, and E.~Segawa.
\newblock Partition-based discrete-time quantum walks.
\newblock {\em Quantum Inf. Process.}, 17(4):100, 2018.

\bibitem{KoShSu}
M.~Kotani, T.~Shirai, and T.~Sunada. 
\newblock Asymptotic behavior of the transition probability of a random walk on an infinite graph. 
\newblock {\em J. Funct. Anal.}, 159:664--689, 1998.

\bibitem{Pet03}
D.~Peterson.
\newblock Gridline graphs: a review in two dimensions and an extension to
  higher dimensions.
\newblock {\em Discrete Appl. Math.}, 126(2):223--239, 2003.

\bibitem{PP17}
P.~Philipp and R.~Portugal.
\newblock Exact simulation of coined quantum walks with the continuous-time
  model.
\newblock {\em Quantum Inf. Process.}, 16(1):14, 2017.

\bibitem{Por16b}
R.~Portugal.
\newblock Staggered quantum walks on graphs.
\newblock {\em Phys. Rev. A}, 93:062335, 2016.

\bibitem{POM17}
R.~Portugal, M.~C. de~Oliveira, and J.~K. Moqadam.
\newblock Staggered quantum walks with {H}amiltonians.
\newblock {\em Phys. Rev. A}, 95:012328, 2017.

\bibitem{PSFG16}
R.~Portugal, R.~A.~M. Santos, T.~D. Fernandes, and D.~N. Gon{\c{c}}alves.
\newblock The staggered quantum walk model.
\newblock {\em Quantum Inf. Process.}, 15(1):85--101, 2016.

\bibitem{SKW03}
N.~Shenvi, J.~Kempe, and K.~B. Whaley.
\newblock A quantum random walk search algorithm.
\newblock {\em Phys. Rev. A}, 67(5):052307, 2003.

\bibitem{Sh} T.~Shirai.
\newblock The spectrum of infinite regular line graph.
\newblock {\em Trans. Amer. Math. Soc.}, 352:115-132, 2000.

\bibitem{Str06}
F.~W. Strauch.
\newblock Connecting the discrete- and continuous-time quantum walks.
\newblock {\em Phys. Rev. A}, 74(3):030301, 2006.

\bibitem{Sun13}
T. Sunada.
\newblock {\em Topological Crystallography: With a View Towards Discrete Geometric Analysis}. Springer, New York, 2013.

\bibitem{TS12}
T.~Sunada and T.~Tate.
Asymptotic behavior of quantum walks on the line.
\newblock {\em J. Funct. Anal.}, 262, 2608--2645, 2012.

\bibitem{Suz16}
A. Suzuki.
Asymptotic velocity of a position-dependent quantum walk.
\newblock {\em Quantum Inf. Process.}, 15(1):103--119, 2016.

\bibitem{Syo51}
I.~Sy\^ozi.
\newblock Statistics of kagom\'e lattice.
\newblock {\em Progress of Theoretical Physics}, 6(3):306--308, 1951.

\bibitem{Sze04a}
M.~Szegedy.
\newblock Quantum speed-up of {M}arkov chain based algorithms.
\newblock In {\em Proc. 45th Annual IEEE Symposium on Foundations of Computer
  Science}, FOCS '04, pages 32--41, Washington, 2004.


\end{thebibliography}
%\bibliographystyle{unsrt}

\end{document}